\documentclass[12pt,reqno]{amsart}
\usepackage{graphicx,amscd, color,amsmath,amsfonts,amssymb,geometry,amssymb, xfrac,nicefrac}
\usepackage[initials]{amsrefs}

\newtheorem{theorem}{Theorem}[section]

\newtheorem{corollary}[theorem]{Corollary}

\newtheorem{remark}[theorem]{Remark}

\newtheorem{lemma}[theorem]{Lemma}

\newtheorem{definition}[theorem]{Definition}
\numberwithin{equation}{section}

\newcommand{\rank}{\operatorname{rank}}
\newcommand{\Ima}{\operatorname{Im}}

\begin{document}

\title{A triality pattern in entanglement theory}

\author[Cariello ]{D. Cariello}

\address{Faculdade de Matem\'atica, \newline\indent Universidade Federal de Uberl\^{a}ndia, \newline\indent 38.400-902 Ð Uberl\^{a}ndia, Brazil.}\email{dcariello@ufu.br}

\keywords{}

\subjclass[2010]{}

\begin{abstract}

In this work, we present new connections between three types of quantum states: positive under partial transpose states, symmetric with positive coefficients states and  invariant under realignment states.
First, we obtain  a common upper bound for their spectral radii and a result on their  filter normal forms. 
Then we  prove the existence of a lower bound for their ranks and the fact that whenever this bound is attained the states are separable. 
These  connections add new evidence to the pattern that for every proven result for one of these types, there are counterparts for the other two, which is a potential source of information for entanglement theory.

\end{abstract}

\maketitle

\section{Introduction}

The separability problem in quantum theory  \cite{Guhne}  asks for a criterion to distinguish the separable states from the entangled states.  It is known that separable states in $\mathcal{M}_k\otimes \mathcal{M}_m$ form a subset of the positive under partial transpose states (PPT states) and in low dimensions, $km\leq 6$, these two sets coincide \cite{peres, horodeckifamily} solving the problem. However, in larger dimensions, $km> 6$, there are entangled PPT states. In addition, for $k,m$ arbitrary, this problem is known to be a hard problem \cite{gurvits2004}, thus any reduction of this problem to a subset of PPT states of $\mathcal{M}_k\otimes \mathcal{M}_m$ is certainly important.  \vspace{0,2cm}

In \cite{cariello, carielloIEEE}, a procedure to reduce the separability problem to a proper subset of PPT states  was presented. The idea behind this reduction can be summarized as follows.  
Let $A=\sum_{i=1}^nA_i\otimes B_i\in \mathcal{M}_k\otimes \mathcal{M}_m$ be a  quantum state and consider $G_A:\mathcal{M}_k\rightarrow \mathcal{M}_m$ and $F_A:\mathcal{M}_m\rightarrow \mathcal{M}_k$ as
 $ G_A(X)=\sum_{i=1}^n tr(A_iX)B_i$\  and\  $ F_A(X)=\sum_{i=1}^n tr(B_iX)A_i$, where $tr(X)$ stands for the trace of $X$.
 
 \vspace{0,2cm}
 
Now, if $A$ is PPT and $X$ is a positive semidefinite Hermitian eigenvector of  $F_A\circ G_A:\mathcal{M}_k\rightarrow \mathcal{M}_k$ then $A$ decomposes  as a sum of states with orthogonal local supports \cite[Lemma 8]{carielloIEEE}, i.e.,
$$A=(V\otimes W)A(V\otimes W)+(V^{\perp}\otimes W^{\perp})A(V^{\perp}\otimes W^{\perp}),$$

\noindent where $V,W, V^{\perp},W^{\perp}$ are orthogonal projections onto $\Ima(X)$, $\Ima(G_A(X))$, $\ker(X)$ and $\ker(G_A(X))$, respectively. 
Then the algorithm proceeds to decompose $(V\otimes W)A(V\otimes W)$ and $(V^{\perp}\otimes W^{\perp})A(V^{\perp}\otimes W^{\perp})$, since they are also PPT, whenever such  $X$ is found. Eventually this process stops with 
\begin{center}
$A=\sum_{i=1}^n(V_i\otimes W_i)A(V_i\otimes W_i)$,
\end{center}
where $(V_i\otimes W_i)A(V_i\otimes W_i)$ cannot be further decomposed for $1\leq i\leq n$. These states are named weakly irreducible. Finally, $A$ is separable  if and only if each $(V_i\otimes W_i)A(V_i\otimes W_i)$  is separable, therefore this algorithm reduces the separability problem to the weakly irreducible PPT case.

\vspace{0,2cm}
Positive under partial transpose states are not the only type of states on which this procedure works because the key feature of this procedure, which is $A\in \mathcal{M}_k\otimes \mathcal{M}_m$ breaks whenever a certain positive semidefinite Hermitian eigenvector is found, is also true for two other types of quantum states.   From now on we say that $A\in \mathcal{M}_k\otimes \mathcal{M}_m$ has the \textbf{completely reducibility property}, if for every positive semidefinite Hermitian eigenvector $X$ of  $F_A\circ G_A:\mathcal{M}_k\rightarrow \mathcal{M}_k$, we have
$$A=(V\otimes W)A(V\otimes W)+(V^{\perp}\otimes W^{\perp})A(V^{\perp}\otimes W^{\perp}),$$

\noindent where $V,W, V^{\perp},W^{\perp}$ are orthogonal projections onto $\Ima(X)$, $\Ima(G_A(X))$, $\ker(X)$ and $\ker(G_A(X))$, respectively. In \cite{carielloIEEE}, a search for other types of quantum states satisfying this property was conducted finding only three types of quantum states: positive under partial transpose states (PPT states), symmetric with positive coefficients states (SPC states) and invariant under realignment states. So far this property has been only verified for these triad of quantum states.

\vspace{0,2cm}

Now, the number of times that $A$ breaks into weakly irreducible pieces is maximized when the non-null eigenvalues of $F_A\circ G_A:\mathcal{M}_k\rightarrow \mathcal{M}_k$ are equal, because in this case we are able to produce many positive semidefinite Hermitian eigenvectors. In this situation, we have the following theorem:

\vspace{0,1cm}

\begin{quote} If the non-null eigenvalues of $F_A\circ G_A$ are equal then $A$ is separable, when $A$ is PPT or SPC or invariant under realignment  \cite[Proposition 15]{carielloIEEE}. Notice that the eigenvalues  of $F_A\circ G_A$ are the square of the Schmidt coefficients of $A$.
\end{quote} 

\vspace{0,1cm}

In \cite{carielloSPC, carielloIEEE}, it was also noticed that every SPC state and every invariant under realignment state is PPT in $\mathcal{M}_2\otimes \mathcal{M}_2$. Before that in \cite{guhnetothsym}, it was proved  that  a state  supported on the symmetric subspace of $\mathbb{C}^k\otimes\mathbb{C}^k$ is PPT  if and only if it is SPC. This is plenty of evidence of how linked this triad of quantum states are. 

\vspace{0,2cm}

In this work we prove new results for this triad of quantum states and a new consequence for their complete reducibility property.  One of these results concerns the separability of theses states. The aforementioned connections along with our new results lead us to notice a certain triality pattern:

\vspace{0,1cm}

 \begin{quote}
For each proven theorem for one of these three types of states, there are corresponding counterparts for the other two.
\end{quote} 

\vspace{0,1cm}

We believe  that a solution to the separability problem for SPC states or  invariant under realignment states would shed light on bound entanglement. This is our reason for studying these types of states.

\vspace{0,2cm}

Next, we would like to point out the origin of these connections which is also the source of the main tools used in this work. For this, we must consider the group of linear contractions and some of its  properties.   For each permutation $\sigma\in S_4$,  the linear transformation $L_{\sigma}:\mathcal{M}_k\otimes \mathcal{M}_k\rightarrow \mathcal{M}_k \otimes \mathcal{M}_k$ satisfying  $$L_{\sigma}(v_1v_2^t\otimes v_3v_4^t)=v_{\sigma(1)}v_{\sigma(2)}^t\otimes v_{\sigma(3)}v_{\sigma(4)}^t,$$
where $v_i\in \mathbb{C}^{k}$, is called a linear contraction.  

\vspace{0,2cm}

The term contraction comes from the fact that  $\|L_{\sigma}(\gamma)\|_1\leq \|\gamma\|_1,$ for every $\sigma\in S_4$, whenever $\gamma\in \mathcal{M}_k\otimes \mathcal{M}_m$ is a separable state and $\|\cdot\|_1$ is the trace norm of a matrix (i.e., the sum of its singular values). 
 Hence, if $\gamma\in \mathcal{M}_k\otimes \mathcal{M}_k$ is a state and $\|L_{\sigma}(\gamma)\|_1> \|\gamma\|_1$ for some $\sigma\in S_4$ then $\gamma\in \mathcal{M}_k\otimes \mathcal{M}_m$ is  entangled. This  observation provides two  useful criteria for entanglement detection. In cycle notation, they are:
\vspace{0,2cm}

 \begin{itemize}
 \item the PPT criterion  \cite{peres, horodeckifamily}, when $\sigma=(34)$, and
 \item the CCNR criterion \cite{rudolph,rudolph2}, when $\sigma=(23)$ . 
\end{itemize}

\vspace{0,2cm}

Despite the name - contraction - these linear maps are in fact isometries, i.e., they preserve the Frobenius norm of $\gamma \in \mathcal{M}_k \otimes \mathcal{M}_k$ $($denoted here by $\|\gamma\|_2)$. This norm is an invariant exploited several times in this work. 
\vspace{0,2cm}

In addition, the set of linear contractions is a group under composition generated by the partial transposes ($L_{(34)}(\gamma)=\gamma^{\Gamma}$ and $L_{(12)}$)  and the realignment map $(L_{(23)}(\gamma)=\mathcal{R}(\gamma))$.  The relations among the elements of this group  are  extremely useful  as we shall see in the proofs of our novel results.

\vspace{0,2cm}

Finally, these  maps allow  connecting this triad of quantum states from their origins, that is, their definitions:
\vspace{0,1cm}

\begin{itemize}
\item PPT states are the states that remain positive under partial transpose $(\gamma\geq 0, \gamma^{\Gamma}\geq 0)$ \cite{peres,horodeckifamily},

\vspace{0,1cm}

\item SPC states are the states that remain positive under partial transpose composed with realignment  $(\gamma\geq 0,\ \mathcal{R}(\gamma^{\Gamma})\geq 0)$ (\cite[corollary 25]{carielloIEEE} and \cite[definition 17]{carielloIEEE}),

\vspace{0,1cm}

\item invariant under realignment states  are the states that remain the same under  realignment $(\gamma\geq 0,\ \mathcal{R}(\gamma)=\gamma)$.
\end{itemize}

\vspace{0,2cm}

These observations on the group of  linear contractions are used  throughout this paper, they are the keys to obtain our novel results. 
Now, let us describe these results.

\vspace{0,2cm}

 Our first  result is an upper bound on the spectral radius for this triad of states.
We show that if $\gamma=\sum_{i=1}^m C_i\otimes D_i\in \mathcal{M}_{k}\otimes \mathcal{M}_k$ is PPT or SPC or invariant under realingment then $$\|\gamma\|_{\infty}\leq \min\{\|\gamma_A\|_{\infty}, \|\gamma_B\|_{\infty}, \|\mathcal{R}(\gamma)\|_{\infty} \},$$
where $\|.\|_{\infty}$ is the operator norm, $\gamma_A=\sum_{i=1}^m C_i tr(D_i)$ and $\gamma_B=\sum_{i=1}^m D_i tr(C_i)$ are the reduced states. Let us say that the ranks of $\gamma_A$ and $\gamma_B$ are the reduced ranks of $\gamma$.

\vspace{0,2cm}

 Our second  result regards the filter normal form of SPC states and  invariant under realignment states. This normal form has been used in entanglement theory, for example, to provide a different solution for the separability problem in $\mathcal{M}_2\otimes \mathcal{M}_2$ \cite{leinaas} or to prove the equivalence of some criteria for entanglement detection \cite{Git}.   Here we show that  states which are SPC or  invariant under realignment  can be put in the filter normal form and their normal forms can still be chosen to be SPC and invariant under realignment, respectively. In other words, if $A,B\in\mathcal{M}_{k}\otimes \mathcal{M}_k$ are states such that
 
 \vspace{0,2cm}
 
\begin{enumerate}
\item $A$ is SPC then there is an invertible matrix $R\in \mathcal{M}_k$ such that $(R\otimes R)A(R\otimes R)^*= \sum_{i=1}^n\lambda_i\delta_i\otimes\delta_i,$ where $\lambda_1=\frac{1}{k}$,  $\delta_1=\frac{Id}{\sqrt{k}}$, $\lambda_i>0$ and $\delta_i$ is Hermitian for every $i$, and $tr(\delta_i\delta_j)=0$ for $i\neq j$.
\vspace{0,2cm}

\item  $B$  is invariant under realignment then there is an invertible matrix $R\in \mathcal{M}_k$ such that $(R\otimes \overline{R})B(R\otimes \overline{R})^*= \sum_{i=1}^n\lambda_i\delta_i\otimes\overline{\delta_i},$ where $\lambda_1=\frac{1}{k}$,  $\delta_1=\frac{Id}{\sqrt{k}}$, $\lambda_i>0$ and $\delta_i$ is Hermitian for every $i$, and $tr(\delta_i\delta_j)=0$ for $i\neq j$.
\end{enumerate} 

\vspace{0,2cm}
 
  The PPT counterpart of this result  is an algorithm that determines whether a PPT state can be put in the filter normal form or not already found in \cite{CarielloLMP}. This algorithm is based on the complete reducibility property. 
  
  \vspace{0,2cm}
 
 Our final result is  a lower bound for the ranks of these three types of states. We show that the rank of PPT states, SPC states and  invariant under realignment states cannot be inferior to their reduced ranks (when they are equal) and whenever this minimum is attained the states are separable. 
 
 \vspace{0,2cm}
 
 In \cite{PawelMxN}, it was proved that a state $\gamma$ such that  $\rank(\gamma)\leq \max\{\rank(\gamma_A),\rank(\gamma_B)\}$ is separable by just being positive under partial tranpose. In  \cite{smolin},  it was shown that the rank of a separable state is greater or equal to their reduced ranks. So if $\gamma$ is PPT and  $\rank(\gamma)\leq \max\{\rank(\gamma_A),\rank(\gamma_B)\}$ then it is separable and  $\rank(\gamma)= \max\{\rank(\gamma_A),\rank(\gamma_B)\}$.
 
 \vspace{0,2cm}
 
 Hence our last result  is known for PPT states, but  it is original for SPC states and  invariant under realignment states. The approach presented here to obtain these facts is completely  original, as we show that this is another consequence of the complete reducibility property. 

\vspace{0,2cm}

 The results described above show how fundamental this property is to entanglement  theory, it acts as a unifying  approach for many results. In fact, even outside entanglement theory, we can find consequences of that property, for example, a new proof of Weiner's theorem  \cite{weiner} on mutually unbiased bases  found in \cite{carielloIEEE}.
\vspace{0,2cm}
 
  The triality pattern described above together with the complete reducibility property form a potential source of information for entanglement theory.
 
\vspace{0,2cm}

This paper is organized as follows. In section 2, we present some preliminary results, which are mainly facts about the group of linear contractions. In section 3, we obtain an upper bound for the spectral radius of our special triad of quantum states. In section 4, we show that SPC states and invariant under realignment states can be put in the filter normal form and their normal forms retain their shape. In section 5, we prove that the ranks of our triad of quantum states cannot be inferior to their reduced ranks and whenever this minimum is attained the states are separable.

\section{Preliminary results}

In this section we present some preliminary results. We begin by noticing that 
$G_A:\mathcal{M}_k\rightarrow \mathcal{M}_m$ and $F_A:\mathcal{M}_m\rightarrow \mathcal{M}_k$
defined in the introduction are adjoint maps with respect to the trace inner product, when $A$ is Hermitian. The reason behind this is quite simple: If $A$ is Hermitian then $F_A(X)^*=F_{A^*}(X^*)=F_A(X^*)$, for every $X\in \mathcal{M}_m$, hence 
 $$tr(G_A(X)Y^*)=tr(A(X\otimes Y^*))=tr(XF_A(Y^*))=tr(XF_A(Y)^*).$$
 
Notice that  for positive semidefinite Hermitian matrices $X\in\mathcal{M}_k,Y\in\mathcal{M}_m$ and $A\in \mathcal{M}_k\otimes\mathcal{M}_m$, we have $tr(A(X\otimes Y^*))\geq 0$. Thus, the equality above also shows that $G_A$ and $F_A$ are positive maps (Definition \ref{definitionpositivemaps}) when $A$ is positive semidefinite.\vspace{0,3cm}

These maps are connected to the following generalization of the Hadamard product extensively used in \cite{cariello} and required here a few times.

\begin{definition}[Generalization of the Hadamard product]\label{definitionproduct}
Let $\gamma=\sum_{i=1}^nA_i\otimes B_i\in \mathcal{M}_m\otimes M_k$, $\delta=\sum_{j=1}^lC_j\otimes D_j\in \mathcal{M}_k\otimes M_s$. 
Define the product $\gamma*\delta\in  \mathcal{M}_m\otimes \mathcal{M}_s$ as $$\gamma*\delta=\sum_{i=1}^n\sum_{j=1}^lA_i\otimes D_j tr(B_iC_j^t).$$
\end{definition} 
\vspace{0,1cm}

Let us  recall some facts regarding this product.

  \vspace{0,1cm}
\begin{remark}\label{remarkproduct} Let $u=\sum_{i=1}^ke_i\otimes e_i\in\mathbb{C}^k\otimes \mathbb{C}^k$, where $e_1,\ldots,e_k$ is the canonical basis of $\mathbb{C}^k$. 

\begin{itemize}
\item[$a)$] Notice that $u^t(B_i\otimes C_j) u=tr((B_i\otimes C_j) uu^ t)=tr(B_iC_j^t)$, where $B_i,C_j\in\mathcal{M}_k$. Therefore, $$\gamma*\delta=(Id_{m\times m}\otimes u^t\otimes Id_{s\times s})(\gamma\otimes \delta)(Id_{m\times m}\otimes u\otimes Id_{s\times s}),$$ 
which implies that $\gamma*\delta$ is positive semidefinite, whenever $\gamma,\delta$ are positive semidefinite. In addition, $tr(\gamma*\delta)=tr(\gamma\otimes \delta\ (Id\otimes uu^t\otimes Id))=tr(\gamma_B\otimes \delta_A\  uu^t)=tr(\gamma_B\delta_A^t).$ \vspace{0,3cm}
\item[$b)$] By \cite[Proposition 8]{cariello},  $\gamma*\delta=(F_{\gamma}((\cdot)^t)\otimes Id)(\delta)=(Id\otimes G_{\delta}((\cdot)^t))(\gamma)$.\vspace{0,3cm}
\item[$c)$] Let $F=(uu^t)^{\Gamma}\in \mathcal{M}_k\otimes \mathcal{M}_k$ and notice that
if $\gamma=\sum_{i=1}^nA_i\otimes B_i\in \mathcal{M}_k\otimes \mathcal{M}_k$ then $F\gamma F=\sum_{i=1}^n B_i\otimes A_i.$ This real symmetric matrix $F$ is an isometry and it is  usually called the flip operator.\vspace{0,3cm}
\item[] Now, notice that $\gamma*(F\gamma^t F)=\sum_j\sum_{i}A_i\otimes A_j^t tr(B_iB_j)$.  Hence \begin{center} $tr(\gamma*(F\gamma^t F))=tr((\sum_{i}tr(A_i) B_i)(\sum_{j}tr(A_j)B_j))=tr(\gamma_B^2)$ and \vspace{0,2cm}

$G_{\gamma*(F\gamma^t F)}(X)=\sum_j\sum_{i}tr(A_iX)tr(B_iB_j)A_j^t=\sum_jtr((\sum_{i}tr(A_iX)B_i)B_j)A_j^t=F_{\gamma}(G_{\gamma}(X))^t.$
\end{center}
 
\end{itemize}

\end{remark}
 
  \vspace{0,5cm}

Next, we discuss some facts about the group of linear contractions. Actually, we need to focus only on three of these maps.  The linear contractions important to us are

  \vspace{0,2cm}

\begin{center}
$L_{(34)}(\gamma)=\gamma^{\Gamma}$, $L_{(24)}(\gamma)=\gamma F$ and $L_{(23)}(\gamma)=\mathcal{R}(\gamma).$ 
\end{center}

\vspace{0,2cm}

Below we discuss several properties of these linear contractions such as relations among these elements and how they behave with respect to the product defined in  \ref{definitionproduct}.
  \vspace{0,3cm}

\begin{lemma}\label{propertiesofrealignment} Let $\gamma,\delta \in  \mathcal{M}_k\otimes \mathcal{M}_k$ and $v=\sum_{i}a_i\otimes b_i,\ w=\sum_{j} c_j\otimes d_j \in \mathbb{C}^k\otimes \mathbb{C}^k$.
\begin{enumerate}
\item $\mathcal{R}(vw^t)=V\otimes W$, where $V=\sum_{i}a_ib_i^t,\ W=\sum_{j} c_jd_j^t$. 
\item $\mathcal{R}(\mathcal{R}(\gamma))=\gamma$
\item $\mathcal{R}((V\otimes W)\gamma (M\otimes N))=(V\otimes M^t)\mathcal{R}(\gamma)(W^t\otimes N)$
\item $\mathcal{R}(\gamma  F)F=\gamma^{\Gamma}$
\item $\mathcal{R}(\gamma^{\Gamma})=\mathcal{R}(\gamma)F$
\item $\mathcal{R}(\gamma F)=\mathcal{R}(\gamma)^{\Gamma}$
\item $\mathcal{R}(\gamma^{\Gamma})^{\Gamma}=\gamma F$
\item $\mathcal{R}(\gamma*\delta)=\mathcal{R}(\gamma)\mathcal{R}(\delta)$ $($i.e., $\mathcal{R}$ is a homomorphism$)$.
\item $\mathcal{R}(F\overline{\gamma}F)=\mathcal{R}(\gamma)^*$
\end{enumerate}
\end{lemma}
\begin{proof}
Items (1--6) were proved in items (2--7) of \cite[Lemma 23]{carielloIEEE}. For the other three items, we just need to prove them when $\gamma=ab^t\otimes cd^t$ and $\delta=ef^t\otimes gh^t$, where $a,b,c,d,e,f,g,h\in \mathbb{C}^k$.\\\\
\underline{Item (7):} $\mathcal{R}(ab^t\otimes cd^t)^{\Gamma})^{\Gamma}=\mathcal{R}(ab^t\otimes dc^t)^{\Gamma}=(ad^t\otimes bc^t)^{\Gamma}=ad^t\otimes cb^t$.

Now, $(ab^t\otimes cd^t)F=ad^t\otimes cb^t$. So $\mathcal{R}(\gamma^{\Gamma})^{\Gamma}=\gamma F$.\\\\
\underline{Item (8):} $\mathcal{R}(ab^t\otimes cd^t* ef^t\otimes gh^t)=\mathcal{R}(ab^t\otimes gh^t)(d^tf)(c^te)=(ag^t\otimes bh^t)(d^tf)(c^te)$.

Now, $\mathcal{R}(ab^t\otimes cd^t)\mathcal{R}(ef^t\otimes gh^t)=(ac^t\otimes bd^t)(eg^t\otimes fh^t)=(ag^t\otimes bh^t)(c^te)(d^tf)$. 

So $\mathcal{R}(\gamma*\delta)=\mathcal{R}(\gamma)\mathcal{R}(\delta)$\\\\
\underline{Item (9):} $\mathcal{R}(F\overline{a}\overline{b}^t\otimes \overline{c}\overline{d}^t F)=\mathcal{R}(\overline{c}\overline{d}^t\otimes \overline{a}\overline{b}^t)=\overline{c} \overline{a}^t\otimes \overline{d}\overline{b}^t$

Now, $\mathcal{R}(ab^t\otimes cd^t)^*=(ac^t\otimes bd^t)^*=\overline{c}\overline{a}^t\otimes \overline{d}\overline{b}^t.$

So $\mathcal{R}(F\overline{\gamma} F)=\mathcal{R}(\gamma)^*$.
\end{proof}

\vspace{0,3cm}

The next lemma is important for our final result and says something very interesting about PPT states which remain PPT under realignment: They must be invariant under realignment.

\vspace{0,3cm}

\begin{lemma} \label{lemmarealigmentPPTareequal}Let $\gamma \in \mathcal{M}_k\otimes \mathcal{M}_k$ be a positive semidefinite Hermitian matrix. If $\gamma$ and $\mathcal{R}(\gamma)$ are PPT then $\gamma=\mathcal{R}(\gamma)$.
\end{lemma}
\begin{proof}
By item (4) of lemma \ref{propertiesofrealignment}, $\gamma^{\Gamma} =\mathcal{R}(\gamma F)F$. \vspace{0,2cm}

Now, since $F^2=Id$,   $\gamma^{\Gamma} F =\mathcal{R}(\gamma F)$. \vspace{0,2cm}

Next, by item $(6)$ of  lemma \ref{propertiesofrealignment}, $\mathcal{R}(\gamma F)=\mathcal{R}(\gamma)^{\Gamma}$. So $\gamma^{\Gamma}F=\mathcal{R}(\gamma)^{\Gamma}$ is a positive semidefinite Hermitian matrix by hypothesis.\vspace{0,2cm}

Since $F$,  $\gamma^{\Gamma}$ and $\gamma^{\Gamma}F$ are Hermitian matrices, $\gamma^{\Gamma}F=F\gamma^{\Gamma}$.
So there is an orthonormal basis of $\mathbb{C}^{k}\otimes\mathbb{C}^k$ formed by  symmetric and anti-symmetric eigenvectors of $\gamma^{\Gamma}$. \vspace{0,2cm}

Remind that $\gamma^{\Gamma}$ and $\gamma^{\Gamma}F$ are positive semidefinite,  hence $\gamma^{\Gamma}=\gamma^{\Gamma}F$.\vspace{0,2cm}

Finally, we have noticed that $\gamma^{\Gamma}F=\mathcal{R}(\gamma)^{\Gamma}$. So $\gamma^{\Gamma}=\mathcal{R}(\gamma)^{\Gamma}$, which implies $\gamma=\mathcal{R}(\gamma).$
\end{proof}

\vspace{0,3cm}

The next lemma is used in this work a few times. Although simple, we state  it here in order to better organize our arguments.

\vspace{0,3cm}

\begin{lemma}\label{lemmaoperatornormrealignment} Let $\gamma\in \mathcal{M}_k\otimes \mathcal{M}_k$. The largest singular value of the map $G_{\gamma}:\mathcal{M}_k\rightarrow \mathcal{M}_k$ equals $$\|\mathcal{R}(\gamma)\|_{\infty}=\max\{|tr(\gamma (X\otimes Y))|,\ \|X\|_2=\|Y\|_2=1 \}.$$ In addition, if $\gamma$ is Hermitian then there are Hermitian matrices $\gamma_1,\delta_1\in \mathcal{M}_k$ such that $\|\mathcal{R}(\gamma)\|_{\infty}=tr(\gamma (\gamma_1\otimes\delta_1))$, where $\|\gamma_1\|_2=\|\delta_1\|_2=1$.
\end{lemma}
\begin{proof}First of all, recall that $\|\mathcal{R(\gamma)}\|_{\infty}=\max\{|tr(\mathcal{R}(\gamma)vw^t)|,\ v,w\in \mathbb{C}^k\otimes \mathbb{C}^k\text{ are unit vectors}\}.$
\vspace{0.2cm}

Now, since $\mathcal{R}$ is an isometry, $tr(\mathcal{R}(\gamma)vw^t)=tr(\mathcal{R}(\gamma)(\overline{w}\overline{v}^t)^*)=tr(\gamma \mathcal{R}(\overline{w}\overline{v}^t)^*).$ \vspace{0.2cm}

By item (1) of lemma \ref{propertiesofrealignment}, $\mathcal{R}(\overline{w}\overline{v}^t)^*=W^t\otimes V^t$, where $W,V\in \mathcal{M}_k$ and $\|W\|_2= \|V\|_2=1$.
\vspace{0.2cm}

Therefore,  $\|\mathcal{R}(\gamma)\|_{\infty}=\max\{|tr(\gamma (W^t\otimes V^t))|,\ W,V\in \mathcal{M}_k\text{ and }\|W\|_2= \|V\|_2=1\},$\vspace{0,2cm}

\hspace{3.65cm} $=\max\{|tr(G_{\gamma}(W^t)V^t)|,\ W,V\in \mathcal{M}_k\text{ and }\|W\|_2= \|V\|_2=1\},$\vspace{0,2cm}

\hspace{3.65cm} $=\max\{\|(G_{\gamma}(W^t)\|_2,\ W\in \mathcal{M}_k\text{ and }\|W\|_2=1\}.$\vspace{0,2cm}

\hspace{3.65cm} $=$ the largest singular value of $G_{\gamma}:\mathcal{M}_k\rightarrow \mathcal{M}_k$. \vspace{0,2cm}

This proves the first part of the lemma. Now for the second part, if $\gamma$ is Hermitian then  the set of Hermitian matrices is left invariant by $G_{\gamma}:\mathcal{M}_k\rightarrow \mathcal{M}_k$. Therefore, there is an Hermitian matrix $\gamma_1\in \mathcal{M}_k$ such that $\|\gamma_1\|_2=1$  and $\|G_{\gamma}(\gamma_1)\|_2=$ the largest singular value of $G_{\gamma}$.\vspace{0.2cm}

Notice that  $\|G_{\gamma}(\gamma_1)\|_2=tr(G_{\gamma}(\gamma_1)\delta_1)=tr(\gamma(\gamma_1\otimes\delta_1))$, where $\delta_1=G_{\gamma}(\gamma_1)/\|G_{\gamma}(\gamma_1)\|_2$. \vspace{0.2cm}
 
So $\|\mathcal{R}(\gamma)\|_{\infty}=tr(\gamma(\gamma_1\otimes\delta_1))$, where $\|\gamma_1\|_2=\|\delta_1\|_2=1$ and $\gamma_1$, $\delta_1$ are Hermitian matrices.
\end{proof}

\vspace{0,2cm}

Now, we have all the preliminary results required to discuss our new results.

\vspace{0,2cm}

\section{An upper bound for the spectral radius of the special triad}

In this section we obtain an upper bound for the spectral radius of PPT states, SPC states and invariant under realignment states (theorem \ref{specialspectralradius}). In order to prove this theorem, two lemmas are required.

\vspace{0,2cm}

\begin{lemma}\label{generalspectralradius} Let  $\gamma\in \mathcal{M}_k\otimes \mathcal{M}_k$ be any positive semidefinite Hermitian matrix. Then 
$$\|\gamma^{\Gamma}\|_{\infty}\leq \min\left\{\|\gamma_A\|_{\infty}, \|\gamma_B\|_{\infty}, \|\mathcal{R}(\gamma)\|_{\infty} \right\}.$$

\end{lemma}

\begin{proof} Let $v\in \mathbb{C}^k\otimes\mathbb{C}^k$ be a unit vector such that  $|tr(\gamma^{\Gamma}vv^*)|=\|\gamma^{\Gamma}\|_{\infty}$. Let $n$ be its rank.
Denote by  $\{g_1,\ldots, g_k\}$ and $\{e_1,\ldots,e_n\}$  the canonical bases of $\mathbb{C}^k$ and $\mathbb{C}^n$, respectively.
\vspace{0,2cm}

Next, there are matrices $D\in \mathcal{M}_{k\times k}$, $E\in \mathcal{M}_{k\times k}$,  $R\in \mathcal{M}_{k\times n}$ and $S\in \mathcal{M}_{k\times n}$ such that
\vspace{0,1cm}

\begin{enumerate}
\item $v=(D\otimes Id) w$, where $tr(DD^*)=1$, $w=\sum_{i=1}^k g_i\otimes g_i\in \mathbb{C}^k\otimes\mathbb{C}^k $,

\vspace{0,1cm}
\item $v=(Id\otimes E) w$, where  $tr(\overline{E}E^t)=1$, $w=\sum_{i=1}^k g_i\otimes g_i\in \mathbb{C}^k\otimes\mathbb{C}^k$,

\vspace{0,1cm}

\item $v=(R\otimes S)u$, where $u=\sum_{i=1}^ne_i\otimes e_i\in \mathbb{C}^n\otimes \mathbb{C}^n$ and  $tr((RR^*)^2)=tr((\overline{S}S^t)^2 )=1$.
\end{enumerate}

\vspace{0,3cm}
Now,

\begin{enumerate}
\item $\|\gamma^{\Gamma}\|_{\infty}=|tr(\gamma^{\Gamma}vv^*)|=|tr((D^*\otimes Id)\gamma(D\otimes Id)(ww^*)^{\Gamma})|\leq tr((D^*\otimes Id)\gamma(D\otimes Id)), $\\ since $Id\otimes Id \pm (ww^*)^{\Gamma}$ and $\gamma$ are  positive semidefinite.  Hence $$\|\gamma^{\Gamma}\|_{\infty}\leq  tr(\gamma(DD^*\otimes Id))=tr(\gamma_ADD^*)\leq \|\gamma_A\|_{\infty}tr(DD^*)=\|\gamma_A\|_{\infty}.$$

\vspace{0,2cm}

\item $\|\gamma^{\Gamma}\|_{\infty}=|tr(\gamma^{\Gamma}vv^*)|=|tr((Id\otimes E^t)\gamma(Id\otimes \overline{E})(ww^*)^{\Gamma})|\leq tr((Id\otimes E^t)\gamma(Id\otimes \overline{E})), $ \\ since $Id\otimes Id \pm (ww^*)^{\Gamma}$ and $\gamma$ are  positive semidefinite. Hence $$\|\gamma^{\Gamma}\|_{\infty}\leq  tr(\gamma(Id\otimes \overline{E}E^t))=tr(\gamma_B\overline{E}E^t)\leq \|\gamma_B\|_{\infty}tr(\overline{E}E^t)=\|\gamma_B\|_{\infty}.$$

\vspace{0,2cm}

\item $\|\gamma^{\Gamma}\|_{\infty}=|tr(\gamma^{\Gamma}vv^*)|=|tr((R^*\otimes S^t)\gamma(R\otimes \overline{S})(uu^*)^{\Gamma})|\leq tr((R^*\otimes S^t)\gamma(R\otimes \overline{S}))$, \\ since $Id\otimes Id \pm (uu^*)^{\Gamma}$ and $\gamma$ are  positive semidefinite. Hence $$\|\gamma^{\Gamma}\|_{\infty}\leq  tr(\gamma(RR^*\otimes \overline{S}S^t))\leq \|\mathcal{R}(\gamma)\|_{\infty},\text{ since }\|RR^*\|_2=\|\overline{S}S^t \|_2=1, \text{ by lemma } \ref{lemmaoperatornormrealignment}.$$
\end{enumerate}

\end{proof}

\vspace{0,5cm}

\begin{lemma}\label{generalspectralradiusrealignment} Let  $\gamma\in \mathcal{M}_k\otimes \mathcal{M}_k$ be a positive semidefinite Hermitian matrix.  
Then 
$$\|\mathcal{R}(\gamma)\|_{\infty}^2\leq \|\gamma_A\|_{\infty}\|\gamma_B\|_{\infty}.$$
\end{lemma}

\begin{proof}
By lemma \ref{lemmaoperatornormrealignment}, there are Hermitian matrices $\gamma_1\in \mathcal{M}_k$ and $\delta_1\in \mathcal{M}_k$ such that \\

\begin{enumerate}
\item $tr(\gamma_1^2)=tr(\delta_1^2)=1$\\

\item $tr(\gamma(\gamma_1\otimes \delta_1))=\|\mathcal{R}(\gamma)\|_{\infty}.$\\
\end{enumerate}

Consider the following positive semidefinite Hermitian matrix

$$\begin{pmatrix}
\gamma^{\frac{1}{2}} & 0\\
0 & \gamma^{\frac{1}{2}}
\end{pmatrix}\begin{pmatrix}
Id\otimes \delta_1\\
 \gamma_1\otimes Id 
\end{pmatrix}
\begin{pmatrix}
Id\otimes \delta_1 & \gamma_1\otimes Id
\end{pmatrix}\begin{pmatrix}
\gamma^{\frac{1}{2}} & 0\\
0 & \gamma^{\frac{1}{2}}
\end{pmatrix}$$
\vspace{0,5cm}
$$=
\begin{pmatrix}
\gamma^{\frac{1}{2}}(Id\otimes \delta_1^2)\gamma^{\frac{1}{2}} & \gamma^{\frac{1}{2}}(\gamma_1\otimes\delta_1)\gamma^{\frac{1}{2}} \\
\gamma^{\frac{1}{2}}(\gamma_1\otimes \delta_1)\gamma^{\frac{1}{2}} & \gamma^{\frac{1}{2}}(\gamma_1^2\otimes Id)\gamma^{\frac{1}{2}}
\end{pmatrix}.$$

\vspace{0,5cm}

Its  partial trace,  \ \ \
$D=\begin{pmatrix}
tr(\gamma (Id_k\otimes \delta_1^2) )& tr(\gamma(\gamma_1\otimes\delta_1)) \\
tr(\gamma(\gamma_1\otimes \delta_1)) & tr(\gamma(\gamma_1^2\otimes Id_k))
\end{pmatrix}_{2\times 2}$
is also positive semidefinite.

\vspace{0,5cm}

Thus \ \ \ $0\leq \det(D)=tr(\gamma (Id_k\otimes \delta_1^2) )tr(\gamma(\gamma_1^2\otimes Id_m))-tr(\gamma(\gamma_1\otimes \delta_1))^2.$\\

Notice that  
\vspace{0,2cm}

\begin{itemize}
\item $tr(\gamma (Id_k\otimes \delta_1^2) )=tr(\gamma_B\delta_1^2)\leq \|\gamma_B\|_{\infty}tr(\delta_1^2)= \|\gamma_B\|_{\infty}$,
\item $tr(\gamma(\gamma_1^2\otimes Id_m))=tr(\gamma_A\gamma_1^2)\leq \|\gamma_A\|_{\infty}tr(\gamma_1^2)= \|\gamma_A\|_{\infty}$,
\item $tr(\gamma(\gamma_1\otimes \delta_1))^2=\|\mathcal{R}(\gamma)\|_{\infty}^2$.\\

\end{itemize}

Hence $\|\mathcal{R}(\gamma)\|_{\infty}^2\leq \|\gamma_A\|_{\infty}\|\gamma_B\|_{\infty}$.
\end{proof}

\vspace{0,3cm}
These  lemmas imply the first new connection for our special triad of quantum states.

\vspace{0,3cm}

\begin{theorem}\label{specialspectralradius}
Let  $\gamma\in \mathcal{M}_k\otimes \mathcal{M}_k$ be a positive semidefinite Hermitian matrix. If $\gamma$ is PPT or SPC or invariant under realingment then $$\|\gamma\|_{\infty}\leq \min\{\|\gamma_A\|_{\infty}, \|\gamma_B\|_{\infty}, \|\mathcal{R}(\gamma)\|_{\infty} \}.$$
\end{theorem}

\begin{proof} First, let $\gamma$ be a PPT state. Hence  $\gamma^{\Gamma}$ is also a state. \\

Notice that
$(\gamma^{\Gamma})_A=\gamma_A$, $(\gamma^{\Gamma})_B=\gamma_B^t$ and, by lemma \ref{lemmaoperatornormrealignment}, $\|\mathcal{R}(\gamma^{\Gamma})\|_{\infty}=\|\mathcal{R}(\gamma)\|_{\infty}$.\\

By applying lemma \ref{generalspectralradius} on $\gamma^{\Gamma}$, we obtain $$\|\gamma\|_{\infty}\leq \min\{\|(\gamma^{\Gamma})_A\|_{\infty}, \|(\gamma^{\Gamma})_B\|_{\infty}, \|\mathcal{R}(\gamma^{\Gamma})\|_{\infty} \}=\min\{\|\gamma_A\|_{\infty}, \|\gamma_B\|_{\infty}, \|\mathcal{R}(\gamma)\|_{\infty} \}.$$

\vspace{0,3cm}
So the proof of the PPT case is complete.\vspace{0,2cm}

Next, if $\gamma$ is SPC or invariant under realignment then $\gamma_A=\gamma_B$ or $\gamma_A=\gamma_B^t$ by lemma \cite[Corollary 25]{carielloIEEE}. 
Hence, by lemma \ref{generalspectralradiusrealignment}, $\|\mathcal{R}(\gamma)\|_{\infty}\leq \|\gamma_A\|_{\infty}.$\vspace{0,2cm}

 It remains to prove that $\|\gamma\|_{\infty}\leq \|\mathcal{R}(\gamma)\|_{\infty}$, whenever $\gamma$ is SPC or invariant under realignment. Notice that this inequality is trivial for matrices invariant under realignment. Thus, let $\gamma$ be  a SPC state.\vspace{0,2cm}

As defined in the introduction, $\mathcal{R}(\gamma^{\Gamma})$ is positive semidefinite. Applying lemma \ref{generalspectralradius} on $\mathcal{R}(\gamma^{\Gamma})$, we obtain
$$\|\mathcal{R}(\gamma^{\Gamma})^{\Gamma}\|_{\infty}\leq \|\mathcal{R}(\mathcal{R}(\gamma^{\Gamma}))\|_{\infty}.$$

\vspace{0,2cm}

Now, by items (7) and (2) of lemma \ref{propertiesofrealignment},  $\mathcal{R}(\gamma^{\Gamma})^{\Gamma}=\gamma F$ and $\mathcal{R}(\mathcal{R}(\gamma^{\Gamma}))=\gamma^{\Gamma}$, where $F$ is the flip operator.
Therefore 
$$\|\gamma F\|_{\infty}\leq \|\gamma^{\Gamma}\|_{\infty}.$$

\vspace{0,2cm}

Finally,  $\|\gamma^{\Gamma}\|_{\infty}\leq \|\mathcal{R}(\gamma)\|_{\infty}$ by lemma \ref{generalspectralradius}, and $\|\gamma \|_{\infty}=\|\gamma F\|_{\infty}$, since $F$ is an isometry. 
 \end{proof}

\vspace{0,5cm}

\section{Filter normal form for SPC states and  invariant under realignment states}

\vspace{0,5cm}

In this section we show that every SPC state and every  invariant under realignment state can be put in the filter normal form. In addition their filter normal forms can still be chosen to be SPC and invariant under realignment, respectively (corollary \ref{corollaryfilternormalformSPCandINV}).  \vspace{0,2cm}

As described in the introduction, there are applications of this normal form in entanglement theory. Now, it has been noticed that this normal form is connected to an extension of Sinkhorn-Knopp theorem for positive maps \cite{CarielloLAMA, gurvits2004}. This theorem concerns the existence of invertible matrices $R,S$ such that  $R^*T(SXS^*)R$ is doubly stochastic for a positive map $T(X)$  satisfying suitable conditions. So we start this section with  some definitions and  lemmas related to this theorem.

\vspace{0,2cm}

Let $V\in \mathcal{M}_k$ be an orthogonal projection and consider the sub-algebra of $\mathcal{M}_k:$ $V\mathcal{M}_kV=\{VXV,\ X\in \mathcal{M}_k\}$. Let $P_k$ denote the set of positive semidefinite Hermitian matrices of $\mathcal{M}_k$.

\vspace{0,2cm}

\begin{definition}\label{definitionpositivemaps}Let us say that $T:V\mathcal{M}_kV\rightarrow V\mathcal{M}_kV$  is a positive map if $T(X)\in P_k\cap V\mathcal{M}_kV$ for every $X\in P_k\cap V\mathcal{M}_kV$. In addition, we say that a positive map $T:V\mathcal{M}_kV\rightarrow V\mathcal{M}_kV$ is doubly stochastic  if the following equivalent conditions hold
\begin{enumerate}
\item  the matrix $A_{m\times m}$, defined as $A_{ij}=tr(T(v_iv_i^*)w_jw_j^*)$, is doubly stochastic for every choice of orthonormal bases $v_1,\ldots,v_m$ and $w_1,\ldots,w_m$  of $\Ima(V)$, 
\item $T(V)=T^*(V)=V$, where $T^*:V\mathcal{M}_kV\rightarrow V\mathcal{M}_kV$ is the adjoint of $T:V\mathcal{M}_kV\rightarrow V\mathcal{M}_kV$ with respect to the trace inner product.

\end{enumerate}

 \end{definition}

\vspace{0,3cm}

\begin{definition}\label{deffullyindecomposabel}A positive map $T:V\mathcal{M}_kV\rightarrow V\mathcal{M}_kV$ is said to be fully indecomposable if the following equivalent conditions hold
\begin{enumerate}
\item the matrix $A_{m\times m}$, defined as $A_{ij}=tr(T(v_iv_i^*)w_jw_j^*)$, is fully indecomposable \cite{marcus} for every choice of orthonormal bases $v_1,\ldots,v_m$ and $w_1,\ldots,w_m$  of $\Ima(V)$, 
\item $\rank(X)+\rank(Y)<\rank(V)$, whenever $X,Y\in (V\mathcal{M}_kV\cap P_k)\setminus\{0\}$and $tr(T(X)Y)=0$,
\item $\rank(T(X))>\rank(X)$,   $\forall X\in V\mathcal{M}_kV\cap P_k$ such that $0<\rank(X)<\rank (V).$\\
\end{enumerate}

\end{definition}

Below we prove  two lemmas concerning self-adjoint maps with respect to the trace inner product.   

\vspace{0,1cm}

\begin{lemma}\label{lemmaselfadjoint} Let $T:V\mathcal{M}_kV\rightarrow V\mathcal{M}_kV$  be a fully indecomposable self-adjoint  map. There is $R\in V\mathcal{M}_kV$ such that $R^*T(R(\cdot)R^*)R: V\mathcal{M}_kV\rightarrow V\mathcal{M}_kV$ is doubly stochastic.
\end{lemma}
\begin{proof}
Since $T:V\mathcal{M}_kV\rightarrow V\mathcal{M}_kV$ is fully indecomposable, it has total support \cite[Lemma 2.3]{CarielloLAMA} or it has a positive achievable capacity \cite{gurvits2004}. So there are matrices $A,B\in V\mathcal{M}_kV$ such that 
$\rank(A)=\rank(B)=\rank(V)$ and $T_1(X)=B^*T(AXA^*)B$ is doubly stochastic  \cite[Theorem 3.7]{CarielloLAMA}. Notice that $T_1$ still is fully indecomposable.

\vspace{0,1cm}

Now,  $T_2(X)=T_1^*(X)=A^*T(B(\cdot)B^*)A$ is also doubly stochastic and \begin{equation}\label{eqrelacaoT2T1}
T_2(X)=C^*T_1(DXD^*)C,
\end{equation}
 where $C=B^+A$, $D=A^+B$  and $Y^+$ is the pseudo-inverse of $Y$.
\vspace{0,3cm}

Let  $C=EFG^*$ and $D=HLJ^*$ be the SVD decompositions of $C$ and $D$, where
\begin{enumerate}
\item $E=(e_1,\ldots,e_m), G=(g_1,\ldots,g_m), H=(h_1,\ldots,h_m), J=(j_1,\ldots,j_m)\in\mathcal{M}_{k\times m}$ and the columns of each of these matrices form an orthonormal basis of  $\Ima(V)$. \vspace{0,2cm}
\item $F=diagonal(f_1,\ldots,f_m)$, $L=diagonal(l_1,\ldots,l_m)$ and $f_i>0$, $l_i>0$ for every $i$.\vspace{0,3cm}

\end{enumerate}

Next, define $R,S\in\mathcal{M}_{m\times m}$ as \begin{center}
$R_{ik}=tr(T_2(j_ij_i^*)g_kg_k^*)$\ \ and\ \ $S_{ik}=tr(T_1(h_ih_i^*)e_ke_k^*)$.
\end{center}

By equation \ref{eqrelacaoT2T1}, $R_{ik}=l_i^2f_k^2\ S_{ik}$, i.e., $R=L^2SF^2$.\vspace{0,2cm}

Thus, $L^2$, $F^2$ are positive diagonal matrices such that $L^2SF^2$ is doubly stochastic by definition \ref{definitionpositivemaps}.  Recall that $S$ is a fully indecomposable matrix by definition \ref{deffullyindecomposabel}.
\vspace{0,2cm}

Since $S$ is fully indecomposable, by a theorem proved in \cite{Sinkhorn},  the diagonal matrices $L^2$ and $F^2$ such that $L^2SF^2$ is doubly stochastic must be unique up to multiplication by positive numbers, but $Id. S. Id$ is also doubly stochastic. Thus, $L=a^{-2} Id$ and  $F=a^{2} Id$ for some $a>0$. \vspace{0,2cm}

Therefore, $B^{+}A=C=a^2U$, where $U=EG^*$.  Notice that $UVU^*=V$.\vspace{0,2cm}

 In addition, $BB^{+}A=a^2BU$. Since $BB^{+}=V$ and $VA=A$, we obtain $A=a^2BU$.\vspace{0,2cm}

Thus, \begin{center}
$B^*T(A(\cdot)A^*)B=B^*T((a^2B)U(\cdot)U^*(a^2B)^*)B=(aB)^*T((aB)U(\cdot)U^*(aB)^*)(aB)$.
\end{center}

Finally, $(aB)^*T((aB)(\cdot)(aB)^*)(aB)$ is doubly stochastic too, since $V=UVU^*$.
\end{proof}

\vspace{0,1cm}

\begin{lemma} \label{keylemmanormalform} Let $T:V\mathcal{M}_kV\rightarrow V\mathcal{M}_kV$ be a self-adjoint positive map such that $v\notin \ker(T(vv^*))$ for every $v\in\Ima(V)\setminus\{\vec{0}\}$. Then there is $R\in V\mathcal{M}_kV$ such that $R^*T(R(\cdot)R^*)R: V\mathcal{M}_kV\rightarrow V\mathcal{M}_kV$ is doubly stochastic.
\end{lemma}

\begin{proof}
This proof is an induction on the $\rank(V)$. \vspace{0,2cm}

If $\rank(V)=1$ then $V\mathcal{M}_kV=\{\lambda vv^*,\ \lambda\in \mathbb{C}\}$. Thus, $T(vv^*)=\mu vv^*$, where $\mu> 0$ by hypothesis.\vspace{0,2cm}

Define $R=\frac{1}{\sqrt[4]{\mu}}vv^*$. So $R^*T(Rvv^*R^*)R=vv^*$. Thus,  $R^*T(R(\cdot)R^*)R: V\mathcal{M}_kV\rightarrow V\mathcal{M}_kV$ is a self-adjoint doubly stochastic map.\vspace{0,2cm}

Let $\rank(V)=n>1$ and assume the validity of this theorem whenever the rank of the orthogonal projection is less than $n$.\vspace{0,2cm}

Consider all pairs of orthogonal projections $(V_1,W_1)$  such that  \begin{center}
  $V_1,W_1\in V\mathcal{M}_kV\setminus\{0\}$ and  $0=tr(T(V_1)W_1)$.
\end{center}

Since there is no $v\in\Ima(V)\setminus\{\vec{0}\}$ such that $tr(T(vv^*)vv^*)=0$, $\Ima(V_1)\cap \Ima(W_1)=\{\vec{0}\}$. So $$\rank(V_1)+\rank(W_1)\leq\rank(V).$$

If for every aforementioned  pair $(V_1,W_1)$, we have  $\rank(V_1)+\rank(W_1)<\rank(V)$, then $T$ is fully indecomposable by definition \ref{deffullyindecomposabel}.  So the result follows by lemma \ref{lemmaselfadjoint}.\vspace{0,2cm}

Let us assume that there is such a pair $(V_1,W_1)$ satisfying $\rank(V_1)+\rank(W_1)=\rank(V)$.\vspace{0,2cm}

Since  $\Ima(V_1)\cap \Ima(W_1)=\{\vec{0}\}$, there is $S\in V\mathcal{M}_kV$ such that $SV_1S^*=V_1$ and $S(V-V_1)S^*=W_1$. Define $T'(X)=S^*T(SXS^*)S$. Note that $tr(T'(V_1)(V-V_1))=0$.\vspace{0,2cm}

Next, since $T$ is self-adjoint so is $T'$, hence $tr(T'(V-V_1)V_1)=0$.\vspace{0,2cm}

These last two equalities imply that \begin{equation}\label{equationsubinv}
T'(V_1\mathcal{M}_kV_1)\subset V_1\mathcal{M}_kV_1\ \ \text{and}\ \ T'((V-V_1)\mathcal{M}_k(V-V_1))\subset (V-V_1)\mathcal{M}_k(V-V_1).
\end{equation}

Of course the restrictions $T'|_{V_1\mathcal{M}_kV_1}$ and $T'|_{(V-V_1)\mathcal{M}_k(V-V_1)}$ are self-adjoint and there is no $v\in\Ima(V_1)\setminus\{\vec{0}\}$  or $v\in\Ima(V-V_1)\setminus\{\vec{0}\}$ such that $tr(T'(vv^*)vv^*)=0$.\vspace{0,2cm}

By induction hypothesis, there are $R_1\in V_1\mathcal{M}_kV_1$ and $R_2\in (V-V_1)\mathcal{M}_k(V-V_1)$ such that \vspace{0,2cm}
\begin{itemize}
\item $R_1^*T'(R_1(\cdot)R_1^*)R_1: V_1\mathcal{M}_kV_1\rightarrow V_1\mathcal{M}_kV_1$ is doubly stochastic, i.e., 
\begin{equation}\label{equationdoublysub1}
 R_1^*T'(R_1(V_1)R_1^*)R_1=V_1
\end{equation}

\item $R_2^*T'(R_2(\cdot)R_2^*)R_2: (V-V_1)\mathcal{M}_k(V-V_1)\rightarrow (V-V_1)\mathcal{M}_k(V-V_1)$ is doubly stochastic, i.e.,
\begin{equation}\label{equationdoublysub2}
R_2^*T'(R_2(V-V_1)R_2^*)R_2=V-V_1
\end{equation}

\end{itemize}
\vspace{0,2cm}

Set $R=R_1+R_2\in V\mathcal{M}_kV$. Note que $T''(X)=R^*T'(RXR^*)R$ is self-adjoint and 

$$T''(V)=T''(V_1+V-V_1)=T''(V_1)+T''(V-V_1)$$
$$\hspace{1 cm}=R^*T'(RV_1R^*)R+R^*T'(R(V-V_1)R^*)R$$
$$\hspace{1,3 cm} =R^*T'(R_1V_1R_1^*)R+R^*T'(R_2(V-V_1)R_2^*)R$$
\begin{center}

$\hspace{4,8 cm} =R_1^*T'(R_1V_1R_1^*)R_1+R_2^*T'(R_2(V-V_1)R_2^*)R_2$, by equation \ref{equationsubinv},\vspace{0,2cm}

$\hspace{2.2 cm}=V_1+(V-V_1)=V$, by equations \ref{equationdoublysub1} and \ref{equationdoublysub2}.
\end{center}

Hence, $T'':V\mathcal{M}_kV\rightarrow V\mathcal{M}_kV$ is doubly stochastic.
\end{proof}

\vspace{0,5cm}

\begin{corollary}\label{corollaryimportant} Let $A\in \mathcal{M}_k\otimes \mathcal{M}_k$ be a Hermitian matrix such that $G_A:\mathcal{M}_k\rightarrow \mathcal{M}_k$ is a self-adjoint positive map  and $tr(A(vv^*\otimes vv^*))> 0$ for every $v\in\mathbb{C}^k\setminus\{0\}$. There is an invertible matrix $R\in \mathcal{M}_k$ such that $(R^*\otimes R^*)A(R\otimes R)=\sum_{i=1}^n\lambda_i \gamma_i\otimes \gamma_i$, where
\begin{enumerate}
\item $\lambda_1=1$ and $\gamma_1=\frac{Id}{\sqrt{k}}$
\item $\lambda_i\in\mathbb{R}$ and $\gamma_i=\gamma_i^*$ for every $i$,
\item $1\geq|\lambda_i|$ for every $i$,
\item $tr(\gamma_i\gamma_j)=0$ for every $i\neq j$ and $tr(\gamma_i^2)=1$ for every $i$.
\end{enumerate}

\end{corollary}
\begin{proof}
By the definition of $G_A:\mathcal{M}_k\rightarrow \mathcal{M}_k$ (given in the introduction), notice that\begin{center}
 $0<tr(A(vv^*\otimes vv^*))=tr(G_A(vv^*)vv^*)$.
\end{center}

 Hence $v\notin \ker G_A(vv^*)$ for every $v\in\mathbb{C}^k\setminus\{0\}$.

By lemma \ref{keylemmanormalform}, there is an invertible matrix $R\in \mathcal{M}_k$ such that $R^*G_A(RXR^*)R$ is doubly stochastic.

Define $B=(R^*\otimes R^*)A(R\otimes R)$ and notice that $G_B(X)=R^*G_A(RXR^*)R$. Therefore, $G_B$ is a self-adjoint doubly stochastic map, i.e., $G_B(\frac{Id}{\sqrt{k}})=\frac{Id}{\sqrt{k}}$.

Let $\frac{Id}{\sqrt{k}}, \gamma_2,\ldots,\gamma_{k^2}$ be an orthonormal basis of $\mathcal{M}_k$ formed by  Hermitian eigenvectors of the self-adjoint positive map $G_B:\mathcal{M}_k\rightarrow \mathcal{M}_k$ such that 
\begin{itemize}
\item $G_B(\gamma_i)=\lambda_i\gamma_i$, where $|\lambda_i|>0$ for $1\leq i\leq n$
\item $G_B(\gamma_i)=0$, for $i>n$.
\end{itemize}

Since $G_B$ is a positive map satisfying $G_B(Id)=Id$ then its spectral radius is 1 \cite[Theorem 2.3.7]{Bhatia1}. So $|\lambda_i|\leq 1$ for every $i$.\vspace{0.2cm}

Finally, by the definition of $G_B$, $$B=\frac{Id}{\sqrt{k}}\otimes G_B\left(\frac{Id}{\sqrt{k}}\right)+\gamma_2\otimes G_B(\gamma_2)+\ldots+\gamma_{k^2}\otimes G_B(\gamma_{k^2})=\sum_{i=1}^n\lambda_i \gamma_i\otimes \gamma_i.$$
\end{proof}
\vspace{0,5cm}

\begin{corollary}\label{corollaryfilternormalformSPCandINV} Let $\gamma\in \mathcal{M}_k\otimes \mathcal{M}_k$ be a positive semidefinite Hermitian matrix such that $\rank(\gamma_A)=k$. There is an invertible matrix $R\in \mathcal{M}_k$ such that 
\begin{enumerate}
\item $(R^*\otimes R^*)\gamma(R\otimes R)=\sum_{i=1}^n\lambda_i \gamma_i\otimes \gamma_i$, if $\mathcal{R}(\gamma^{\Gamma})$ is positive semidefinite;\vspace{0,2cm}
\item $(R^*\otimes R^t)\gamma(R\otimes \overline{R})=\sum_{i=1}^n\lambda_i \gamma_i\otimes \overline{\gamma_i}$, if $\mathcal{R}(\gamma)$ is positive semidefinite,\vspace{0,2cm}
\end{enumerate}

where
\begin{itemize}
\item[$a)$] $\lambda_1=\frac{1}{k}$ and $\gamma_1=\frac{Id}{\sqrt{k}}$
\item[$b)$] $\frac{1}{k}\geq \lambda_i>0$ and $\gamma_i=\gamma_i^*$ for every $i$,
\item[$c)$] $tr(\gamma_i\gamma_j)=0$ for every $i\neq j$ and $tr(\gamma_i^2)=1$ for every $i$.
\end{itemize}

\end{corollary}
\begin{proof} $(1)$
 If $\gamma$ is a state such that $\mathcal{R}(\gamma^{\Gamma})$ is positive semidefinite then, by  \cite[corollary 25]{carielloIEEE}, $\gamma$ can be written as 
$\gamma=\sum_{i=1}^n a_iB_i\otimes B_i$,
where $a_i>0$, $B_i=B_i^*$ and $tr(B_i^2)=1$ for every i, and $tr(B_iB_j)=0$ for $i\neq j$. \vspace{0,3cm}

Hence $G_{\gamma}(X)=\sum_{i=1}^na_iB_itr(B_iX)$ is a self-adjoint  map with positive eigenvalues $a_1,\ldots,a_n$ and possibly some null eigenvalues. In addition, since $\gamma$ is positive semidefinite, $G_{\gamma}(X)$ is a positive map.

Now, let $v\in\mathbb{C}^k$ be such that \begin{center}
$0=tr(\gamma(vv^*\otimes vv^*))=\sum_{i=1}^na_itr(B_ivv^*)^2$. 
\end{center}

Since $a_i>0$  and $tr(B_ivv^*)\in\mathbb{R}$ for every $i$, $tr(B_ivv^*)=0$ for every $i$. Therefore, \begin{center}
$tr(\gamma_Avv^*)=\sum_{i=1}^n a_itr(B_i)tr(B_ivv^*)=0$.
\end{center}

By hypothesis $\gamma_A$ is positive definite, hence $v=0$. 

So, by corollary \ref{corollaryimportant}, there is a invertible matrix $R$ such that 
\begin{equation}\label{eqSPC}
(R^*\otimes R^*)\gamma(R\otimes R)=\sum_{i=1}^n\lambda_i \gamma_i\otimes \gamma_i
\end{equation}

satisfies the four conditions of that corollary. It remains to show that $\lambda_i>0$ and then we multiply equation (\ref{eqSPC}) by $\frac{1}{k}$  to obtain our desired result.\vspace{0,2cm}

Finally, since $G_{\gamma}$ has only non-negative eigenvalues and $\lambda_1,\ldots,\lambda_n$ are non-null eigenvalues of $R^*G_{\gamma}(RXR^*)R$ (as seen in the proof of of corollary \ref{corollaryimportant}), $\lambda_1,\ldots,\lambda_n$ are positive.\\\\
$(2)$  If $\gamma$ is a state such that $\mathcal{R}(\gamma)$ is positive semidefinite then, by  \cite[Corollary 25]{carielloIEEE}, $\gamma$ can be written as
$\gamma=\sum_{i=1}^n a_iB_i\otimes \overline{B_i}$,
where $a_i>0$, $B_i=B_i^*$  and $tr(B_i^2)=1$ for every i, and $tr(B_iB_j)=0$ for $i\neq j$. \vspace{0,3cm}

Consider $\gamma^{\Gamma}=\sum_{i=1}^n a_iB_i\otimes B_i$  and notice that $G_{\gamma^{\Gamma}}(X)=G_{\gamma}(X)^t$ is also a positive map. Now repeat the proof of item $(1)$ for $\gamma^{\Gamma}$. Hence there is a invertible matrix $R$ such that \begin{equation}\label{eqInvReal}
(R^*\otimes R^*)\gamma^{\Gamma}(R\otimes R)=\sum_{i=1}^n\lambda_i \gamma_i\otimes \gamma_i,
\end{equation}
where $\gamma_i$ and $\lambda_i$ satisfy all the required conditions. Finally, 

$$(R^*\otimes R^t)\gamma(R\otimes \overline{R})=\sum_{i=1}^n\lambda_i \gamma_i\otimes \overline{\gamma_i}.$$
\end{proof}

\begin{corollary}\label{corollaryleftfilter} Let $\gamma\in \mathcal{M}_k\otimes \mathcal{M}_k$ be a state. There is an invertible matrix $R\in \mathcal{M}_k$ such that $(R^*\otimes Id)\gamma(R\otimes Id)=\sum_{i=1}^na_i \gamma_i\otimes \delta_i,$ where
\begin{itemize}
\item[$a)$]  $a_1\geq a_i>0$, for every $1\leq i\leq n$, and $\gamma_1=\frac{Id}{\sqrt{k}}$, 
\item[$b)$]  $\gamma_i=\gamma_i^*$, $\delta_i=\delta_i^*$  for every $i$,
\item[$c)$] $tr(\gamma_i\gamma_j)=tr(\delta_i\delta_j)=0$ for every $i\neq j$ and $tr(\gamma_i^2)=tr(\delta_i^2)=1$.
\end{itemize}
\end{corollary}
\begin{proof}
First, since $\gamma$ is a state, so is $F\overline{\gamma}F$. Therefore, $\gamma*F\overline{\gamma}F$ is positive semidefinite by item $a)$ of remark \ref{remarkproduct}.

Now, by items $(8)$ and $(9)$ of lemma \ref{propertiesofrealignment},
$$\mathcal{R}(\gamma*F\overline{\gamma}F)=\mathcal{R}(\gamma)\mathcal{R}(\gamma)^*.$$

By  item $(2)$ of corollary \ref{corollaryfilternormalformSPCandINV}, there is an invertible matrix $R\in\mathcal{M}_k$ such that

\begin{center}
 $(R^*\otimes R^t)(\gamma*F\overline{\gamma}F)(R \otimes \overline{R})=\sum_{i=1}^n\lambda_i\gamma_i\otimes \gamma_i$,
\end{center}

where
\begin{itemize}
\item[$a)$] $\lambda_1=\frac{1}{k}$ and $\gamma_1=\frac{Id}{\sqrt{k}}$
\item[$b)$] $\frac{1}{k}\geq \lambda_i>0$ and $\gamma_i=\gamma_i^*$ for every $i$,
\item[$c)$] $tr(\gamma_i\gamma_j)=0$ for every $i\neq j$ and $tr(\gamma_i^2)=1$ for every $i$.
\end{itemize}

\vspace{0,3cm}

Define $\delta=(R^*\otimes Id)\gamma(R\otimes Id)$ and notice that $$\delta*F\delta^tF=\delta*F\overline{\delta}F=(R^*\otimes R^t)(\gamma*F\overline{\gamma}F)(R \otimes \overline{R}).$$

Thus, $G_{\delta*F\delta^tF}(\frac{Id}{k})=\lambda_1\frac{Id}{k}$.

\vspace{0,3cm}

By item $c)$ of remark \ref{remarkproduct}, $F_{\delta}(G_{\delta}(\frac{Id}{\sqrt{k}}))=G_{\delta*F\delta^tF}(\frac{Id}{\sqrt{k}})^t=\lambda_1\frac{Id}{\sqrt{k}}$.  So $F_{\delta}(G_{\delta}(Id))=\lambda_1Id$.

\vspace{0,3cm}

By \cite[Theorem 2.3.7]{Bhatia1}, $\lambda_1$ is the largest eigenvalue of the positive map $F_{\delta}\circ G_{\delta}$. So $\sqrt{\lambda_1}$ is the largest singular value of $G_{\delta}$ and $F_{\delta}$, since they are adjoints.

\vspace{0,3cm}

Next, let $\delta_1=\frac{Id}{\sqrt{k}},\delta_2,\ldots,\delta_{k^2}$ be an orthonormal basis of $\mathcal{M}_k$ formed by Hermitian eigenvectors of $F_{\delta}\circ G_{\delta}:\mathcal{M}_k\rightarrow \mathcal{M}_k$. 

\vspace{0,3cm}

Notice that $(R^*\otimes Id)\gamma(R\otimes Id)=\delta=\sum_{i=1}^{k^2}\delta_i\otimes G_{\delta}(\delta_i)$.

\vspace{0,3cm}

If $G_{\delta}(\delta_i)\neq 0_{k\times k}$, for $1\leq i\leq n$, then define $a_i=\|G_{\delta}(\delta_i)\|_2>0$. Thus, $\delta=\sum_{i=1}^{n}a_i\delta_i\otimes \frac{1}{a_i}G_{\delta}(\delta_i)$. 

\vspace{0,3cm}

Notice that $G_{\delta}(\delta_i)^*=G_{\delta^*}(\delta_i^*)=G_{\delta}(\delta_i)$, since $\delta$ and $\delta_i$ are Hermitian matrices. Moreover, by the definition of $a_i$, \begin{center}
$tr(\frac{1}{a_i}G_{\delta}(\delta_i)\frac{1}{a_i}G_{\delta}(\delta_i))=1,\ \ \ 1\leq i\leq n$. 
\end{center}

In addition, since $F_{\delta}(G_{\delta}(\delta_j)))$ is a multiple of $\delta_j$, $\delta_i$ is orthogonal to $F_{\delta}(G_{\delta}(\delta_j)))$ for $i\neq j$,  \begin{center}
$tr(\frac{1}{a_i}G_{\delta}(\delta_i)\frac{1}{a_j}G_{\delta}(\delta_j))=\frac{1}{a_ia_j}tr(\delta_i F_{\delta}(G_{\delta}(\delta_j)))=0$.
\end{center}

Finally,  $a_1^2=tr(G_{\delta}(\delta_1)^2)=tr(\delta_1F_{\delta}(G_{\delta}(\delta_1)))=\lambda_1tr(\delta_1^2)=\lambda_1$, so $a_1=\sqrt{\lambda_1}$. Notice also that, for every $i$,
$a_i\leq$ the largest singular value of $G_{\delta}$, which is $ \sqrt{\lambda_1}=a_1$.

\end{proof}

\vspace{0,5cm}

\section{Lower bound for the rank of the special triad}

\vspace{0,5cm}

In this section, we prove that $\rank(\gamma)\geq k$, whenever a state $\gamma$ is PPT or  SPC or invariant under realignment and $ \rank(\gamma_A)=\rank(\gamma_B)=k$. Then we show that if $\rank(\gamma)= k$ then $\gamma$ is separable in each of these cases. 

\vspace{0,2cm}

We start by proving  the SPC and the invariant under realignment cases. In their  proofs we use the result that guarantees their separability whenever their Schmidt coefficients are equal, which follows from the complete reducibility property  \cite[Proposition 15]{carielloIEEE}. 

\vspace{0,2cm}

\subsection*{First Cases:} SPC states and  invariant under realignment states

\vspace{0,3cm}

Before proving the inequality, notice that by the symmetry of their Schmidt decompositions \cite[Corollary 25]{carielloIEEE}, $\gamma_B=\gamma_A$ or $\gamma_B=\overline{\gamma_A}$,
 when $\gamma$ is SPC  or invariant under realignment, respectively. Hence $\rank(\gamma_A)=\rank(\gamma_B)$. In addition, we can assume without loss of generality that $\rank(\gamma_A)=k$, otherwise we would be able to embed $\gamma$ in 
$\mathcal{M}_s\otimes \mathcal{M}_s$, where $s=\rank(\gamma_A)$, and obtain the same result.

\begin{theorem}
If $\gamma\in \mathcal{M}_k\otimes \mathcal{M}_k$ is a SPC state such that $\rank(\gamma_A)=k$ then $\rank(\gamma)\geq k$.  In addition, if the equality holds then $\gamma$ is separable. 
\end{theorem}
\begin{proof} 
By corollary \ref{corollaryfilternormalformSPCandINV}, there is a invertible matrix $R$ such that   \begin{center}
$\delta=(R^*\otimes R^*)\gamma(R\otimes R)=\sum_{i=1}^n\lambda_i \gamma_i\otimes \gamma_i$,
\end{center} 

where $\lambda_1=\dfrac{1}{k}$, $\gamma_1=\dfrac{Id}{\sqrt{k}}$, $tr(\gamma_i\gamma_1)=\dfrac{tr(\gamma_i)}{\sqrt{k}}=0$ for $i>1$. So $\displaystyle \delta_A=\sum_{\i=1}^n\lambda_i\gamma_itr(\gamma_i)=\frac{Id}{k}$.

Now, notice that $\delta$ still is SPC by \cite[Corollary 25]{carielloIEEE}. Hence $\|\delta\|_{\infty}\leq \|\delta_A\|_{\infty}=\frac{1}{k}$, by theorem \ref{specialspectralradius}.

So $\dfrac{1}{k}\geq \|\delta\|_{\infty}\geq \dfrac{tr(\delta)}{\rank(\delta)}=\dfrac{1}{\rank(\delta)}$. Hence $ \rank(\delta)\geq k$.   \vspace{0,2cm}

Since $R$ is invertible, $\rank(\gamma)=\rank(\delta)\geq k$. \vspace{0,3cm}

For the next part, assume $\rank(\gamma)=k$. Then $\rank(\delta)=k$. Therefore, \begin{center}
$1=tr(\delta)\leq \|\delta\|_{\infty}\rank(\delta)=\dfrac{1}{k}.k=1$.
\end{center}

Since the equality  $tr(\delta)=\|\delta\|_{\infty}\rank(\delta)$  holds,   the non-null eigenvalues of $\delta$ are equal to $\|\delta\|_{\infty}$. So $tr(\delta)=k\|\delta\|_{\infty}=1$. Hence $\|\delta\|_{\infty}=\frac{1}{k}$ and $tr(\delta^2)=\frac{1}{k}$.
\vspace{0,3cm}

Next, since the linear contraction - $R((\cdot)^{\Gamma})$ - preserves the Frobenius norm of $\delta$, 
 
\begin{equation}\label{eqeigenvalueequal}
\dfrac{1}{k}=tr(\delta^2)=tr(\mathcal{R}(\delta^{\Gamma})\mathcal{R}(\delta^{\Gamma})^*)\leq \|\mathcal{R}(\delta^{\Gamma})\|_{\infty} \|\mathcal{R}(\delta^{\Gamma})\|_{1}.
\end{equation}

\vspace{0,3cm}

By item (5) of lemma \ref{propertiesofrealignment}, $\mathcal{R}(\delta^{\Gamma})=\mathcal{R}(\delta)F$. Since $F$ is an a isometry, $\|\mathcal{R}(\delta)F\|_{\infty}=\|\mathcal{R}(\delta)\|_{\infty}$. 

Therefore,$$\|\mathcal{R}(\delta^{\Gamma})\|_{\infty}=\|\mathcal{R}(\delta)F\|_{\infty}=\|\mathcal{R}(\delta)\|_{\infty}=\lambda_1=\dfrac{1}{k}.$$

Now, since $\delta$ is SPC, by its definition, $\mathcal{R}(\delta^{\Gamma})$ is positive semidefinite. Hence $$\|\mathcal{R}(\delta^{\Gamma})\|_{1}\leq \|\mathcal{R}(\delta^{\Gamma})^{\Gamma}\|_{1}.$$

\vspace{0,3cm}

By item (7) of lemma \ref{propertiesofrealignment}, $\mathcal{R}(\delta^{\Gamma})^{\Gamma}=\delta F$. Thus, $\|\mathcal{R}(\delta^{\Gamma})\|_{1}\leq \|\mathcal{R}(\delta^{\Gamma})^{\Gamma}\|_{1}=\|\delta F\|_{1}=\|\delta\|_{1}=1.$
\vspace{0,3cm}

Using these pieces of information in equation (\ref{eqeigenvalueequal}) we obtain 
$$\dfrac{1}{k}=tr(\mathcal{R}(\delta^{\Gamma})\mathcal{R}(\delta^{\Gamma})^*)\leq \|\mathcal{R}(\delta^{\Gamma})\|_{\infty} \|\mathcal{R}(\delta^{\Gamma})\|_{1}\leq  \dfrac{1}{k}.1.$$

Again, $tr(\mathcal{R}(\delta^{\Gamma})^2)= \|\mathcal{R}(\delta^{\Gamma})\|_{\infty} \|\mathcal{R}(\delta^{\Gamma})\|_{1}$ only holds if the non-null eigenvalues of the positive semidefinite hermitian matrix $\mathcal{R}(\delta^{\Gamma})$ are equal to $\|\mathcal{R}(\delta^{\Gamma})\|_{\infty}=\lambda_1=\dfrac{1}{k}$. 

Finally, the non-null eigenvalues of $\mathcal{R}(\delta^{\Gamma})$  are the non-null Schmidt coefficients of $\delta$. Since $\delta$ is SPC and its non-null Schmidt coefficients are equal, $\delta$ is separable by  \cite[Proposition 15]{carielloIEEE}. Since $R$ is invertible, $\gamma$ is separable too.
\end{proof}
\vspace{0,5cm}

The invariant under realignment counterpart is proved next in a similar way with  minor modifications. 

\vspace{0,5cm}

\begin{theorem}
If $\gamma\in \mathcal{M}_k\otimes \mathcal{M}_k$ is an   invariant under realignment  state such that $\rank(\gamma_A)=k$ then $\rank(\gamma)\geq k$.  In addition, if the equality holds then $\gamma$ is separable. 
\end{theorem}
\begin{proof}
By corollary \ref{corollaryfilternormalformSPCandINV}, there is a invertible matrix $R$ such that   \begin{center}
$\delta=(R^*\otimes R^t)\gamma(R\otimes \overline{R})=\sum_{i=1}^n\lambda_i \gamma_i\otimes \overline{\gamma_i}$,
\end{center} 

where $\lambda_1=\dfrac{1}{k}$, $\gamma_1=\dfrac{Id}{\sqrt{k}}$, $tr(\gamma_i\gamma_1)=\dfrac{tr(\gamma_i)}{\sqrt{k}}=0$ for $i>1$. So $\displaystyle \delta_A=\sum_{\i=1}^n\lambda_i\gamma_itr(\gamma_i)=\frac{Id}{k}$.

Now, by item(3) of lemma \ref{propertiesofrealignment},  $$\mathcal{R}(\delta)=\mathcal{R}((R^*\otimes R^t)\gamma(R\otimes \overline{R}))=(R^*\otimes R^t)\mathcal{R}(\gamma)(R\otimes \overline{R})=(R^*\otimes R^t)\gamma(R\otimes \overline{R})=\delta.$$

Thus, $\delta$ is invariant under realignment. Hence $\|\delta\|_{\infty}\leq \|\delta_A\|_{\infty}=\frac{1}{k}$, by theorem \ref{specialspectralradius}.

So $\dfrac{1}{k}\geq \|\delta\|_{\infty}\geq \dfrac{tr(\delta)}{\rank(\delta)}=\dfrac{1}{\rank(\delta)}$. Hence $ \rank(\delta)\geq k$. 

Since $R$ is invertible, $\rank(\gamma)=\rank(\delta)\geq k$. \vspace{0,3cm}

For the next part, assume $\rank(\gamma)=k$. Then $\rank(\delta)=k$. Therefore, \begin{center}
$1=tr(\delta)\leq \|\delta\|_{\infty}\rank(\delta)=\dfrac{1}{k}.k=1$.
\end{center}

Since the equality  $tr(\delta)=\|\delta\|_{\infty}\rank(\delta)$  holds,   the non-null eigenvalues of $\delta$ are equal to $\|\delta\|_{\infty}$. Moreover, $tr(\delta)=k\|\delta\|_{\infty}=1$. Hence $\|\delta\|_{\infty}=\frac{1}{k}$.
\vspace{0,3cm}

Since $\delta$ is invariant under realignment, the  non-null Schmidt coefficients of the Schmidt decomposition of $\delta$ are the the non-null eigenvalues of $\delta$, which are equal. We know that every  invariant under realigment state with equal non-null  Schmidt coefficients is separable by  \cite[Proposition 15]{carielloIEEE}. So $\delta$ is separable and so is $\gamma$, since $R$ is invertible.
\end{proof}

\vspace{0,5cm}

\subsection*{Third case:} The PPT counterpart.

\vspace{0,5cm}

For our final results, we  need some tools developed in sections 2 and 3 together with the complete reducibility property. First, we show that the rank of a PPT state is greater or equal to its reduced ranks in the next lemma.

\vspace{0,5cm}
\begin{lemma}\label{lemmaineqPPT} Let $\gamma\in\mathcal{M}_k\otimes \mathcal{M}_m$ be a PPT state. Then $\rank(\gamma)\geq\max\{\rank(\gamma_A), \rank(\gamma_B)\}.$
\end{lemma}
\begin{proof}
Let us assume without loss of generality that $\max\{\rank(\gamma_A), \rank(\gamma_B)\}=\rank(\gamma_A)=k$. So there is an invertible matrix $R\in\mathcal{M}_k$ such that $R\gamma_AR^*=\frac{1}{k}Id$.

\vspace{0,3cm}

Define $\delta=(R\otimes Id)\gamma(R^*\otimes Id)$ and notice that $\delta_A=\frac{1}{k}Id$.

\vspace{0,3cm}

Since $\delta$ is PPT, by theorem \ref{specialspectralradius}, $\|\delta\|_{\infty}\leq \|\delta_A\|_{\infty}=\frac{1}{k}$. Hence $$1=tr(\delta_A)=tr(\delta)\leq \|\delta\|_{\infty}\rank(\delta)\leq\frac{1}{k} \rank(\delta).$$

Thus, $\rank(\gamma)=\rank(\delta)\geq \max\{\rank(\gamma_A), \rank(\gamma_B)\}$.
\end{proof}

\vspace{0,5cm}

Now, we prove that a PPT state in the filter normal form with minimal rank must be separable which is the key ingredient of the proof of the final theorem.

\vspace{0,5cm}

\begin{lemma}\label{lemmaseparabilityPPT}Let $\gamma\in\mathcal{M}_k\otimes \mathcal{M}_k$ be a PPT state such that
\begin{enumerate}
\item $\gamma_A=\gamma_B=\frac{1}{k}Id$,
\item $\gamma$ has $k$ eigenvalues equal to $\frac{1}{k}$ and the others zero.
\end{enumerate}
Then $\gamma$ is separable.
\end{lemma}
\begin{proof}
This result is trivial in $\mathcal{M}_2\otimes \mathcal{M}_2$, since every PPT state there is separable. Assume the result in true in $\mathcal{M}_i\otimes \mathcal{M}_i$ for $i<k$. Let us prove the result in $\mathcal{M}_k\otimes \mathcal{M}_k$.

\vspace{0,2cm}

Now, since $\gamma$ is a positive semidefinite Hermitian matrix, the linear transformations  $F_{\gamma}$ and  $G_{\gamma}$ are positive maps and adjoints with respect to the trace inner product.

\vspace{0,2cm}

In addition, notice that $\frac{Id}{k}=\gamma_A=F_{\gamma}(Id)$ and $\frac{Id}{k}=\gamma_B=G_{\gamma}(Id)$. Hence $F_{\gamma}(G_{\gamma}(Id))=\frac{1}{k^2}Id$. 

\vspace{0,2cm}

By \cite[Theorem 2.3.7]{Bhatia1},  the spectral radius of the  positive operator $F_{\gamma}\circ G_{\gamma}$ is $\frac{1}{k^2}$. Hence the largest singular value of  $G_{\gamma}$ and $F_{\gamma}$ is $\frac{1}{k}$, since they are adjoints. Thus, \begin{center}
$\|G_{\gamma}(X)\|_2\leq \frac{1}{k}$ and $\|F_{\gamma}(X)\|_2\leq \frac{1}{k}$, whenever $\|X\|_2=1$.
\end{center}

Next, since $\gamma$ has $k$ linearly independent eigenvectors associated to $\frac{1}{k}$, by combining  2 of them we can find an eigenvector $v\in\mathbb{C}^k\otimes\mathbb{C}^k$  associated to $\frac{1}{k}$ such that $\|v\|_2=1$ and $\rank(v)=m<k$.

\vspace{0,2cm}

Notice that there are $R,S\in \mathcal{M}_k$ with rank $m$ such that\begin{center}
 $v=(R\otimes S)u$ $(u$ as defined in remark \ref{remarkproduct}) and $\|RR^*\|_2=\|SS^*\|_2=1$.
\end{center}

Therefore $\frac{1}{k}=tr(\gamma vv^*)$\vspace{0,1cm}

$\hspace{2,25cm}=tr((R^*\otimes S^*)\gamma (R\otimes S)uu^t)$\vspace{0,1cm}

$\hspace{2,25cm}=tr((R^*\otimes S^t)\gamma^{\Gamma} (R\otimes \overline{S})F)$\ $(F$ as defined in $\ref{remarkproduct})$\vspace{0,1cm}

$\hspace{2,25cm}\leq tr((R^*\otimes S^t)\gamma^{\Gamma} (R\otimes \overline{S}))$\ $($since $\gamma^{\Gamma}$ and $Id-F$ are positive semidefinite$)$\vspace{0,1cm}

$\hspace{2,25cm}= tr(\gamma (RR^*\otimes SS^*))$\vspace{0,1cm}

$\hspace{2,25cm}= tr(G_{\gamma}(RR^*)SS^*)$\vspace{0,1cm}

$\hspace{2,25cm} \leq \|G_{\gamma}(RR^*)\|_2\|SS^*\|_2\leq \frac{1}{k}.1$\ $($since the largest singular value of $G_{\gamma}$ is $\frac{1}{k})$

\vspace{0,5cm}

Therefore, all the inequalities above are equalities, which imply 
\begin{enumerate}
\item $G_{\gamma}(RR^*)=\lambda SS^*$ for some $\lambda>0$, since $G_{\gamma}$ is a positive map, and \vspace{0,1cm}
\item $\frac{1}{k}=\|G_{\gamma}(RR^*)\|_2=\lambda\|SS^*\|_2=\lambda$.
\end{enumerate}

\vspace{0,3cm}

Hence $G_{\gamma}(RR^*)=\frac{1}{k} SS^*$.  
Analogously, we get $F_{\gamma}(SS^*)=\frac{1}{k}RR^*$, since $tr(\gamma (RR^*\otimes SS^*))=tr(RR^*F_{\gamma}(SS^*))$. 

\vspace{0,3cm}

Therefore, $F_{\gamma}(G_{\gamma}(RR^*))=\frac{1}{k^2}RR^*$ and  $\rank(RR^*)=m<k$.

\vspace{0,3cm}

Since $\gamma$ is PPT, by the complete reducibility property, 
\begin{equation}\label{eqcomplredutproperty}
\gamma=(V\otimes W)\gamma(V\otimes W)+(V^{\perp}\otimes W^{\perp})\gamma(V^{\perp}\otimes W^{\perp}),
\end{equation}

where $V,W, V^{\perp},W^{\perp}$ are orthogonal projections onto $\Ima(RR^*)$, $\Ima(SS^*)$, $\ker(RR^*)$ and $\ker(SS^*)$, respectively. 
\vspace{0,3cm}

By equation \ref{eqcomplredutproperty} and the definition of $G_{\gamma}$, 

\begin{center}
 $\Ima(G_{\gamma}(V))\subset \Ima(W)$, $\Ima(G_{\gamma}(V^{\perp}))  \subset \Ima(W^{\perp})$ and\vspace{0.2cm}

\hspace{1.2cm}$\Ima(F_{\gamma}(W))\subset \Ima(V)$, $\Ima(F_{\gamma}(W^{\perp}))\subset \Ima(V^{\perp}).\hspace{2cm}$ 
\end{center}

\vspace{0,3cm}

Next, recall that $V+V^{\perp}=W+W^{\perp}=Id$, $VV^{\perp}=WW^{\perp}=0$ and 

\begin{center}
$G_{\gamma}(V)+G_{\gamma}(V^{\perp})=G_{\gamma}(Id)=\frac{1}{k}Id=\frac{1}{k}W+\frac{1}{k}W^{\perp},$
\end{center}

 \begin{center}
$F_{\gamma}(W)+G_{\gamma}(W^{\perp})=G_{\gamma}(Id)=\frac{1}{k}Id=\frac{1}{k}V+\frac{1}{k}V^{\perp}.$
\end{center} 

\vspace{0,3cm}

Therefore 
\begin{center}
$G_{\gamma}(V)=\frac{1}{k}W$, $F_{\gamma}(W)=\frac{1}{k}V$ and  $G_{\gamma}(V^{\perp})=\frac{1}{k}W^{\perp}$, $G_{\gamma}(W^{\perp})=\frac{1}{k}V^{\perp}$.
\end{center}

Now, define \begin{center}
$\gamma_1=\frac{k}{m}(V\otimes W)\gamma(V\otimes W)$ and $\gamma_2=\frac{k}{k-m}(V^{\perp}\otimes W^{\perp})\gamma(V^{\perp}\otimes W^{\perp})$.
\end{center}

\vspace{0,3cm}

Notice that \begin{center}
$(\gamma_1)_A=F_{\gamma_1}(Id)=\frac{k}{m}F_{\gamma}(W)=\frac{1}{m}V$\ \ and\ \  $(\gamma_1)_B=G_{\gamma_1}(Id)=\frac{k}{m}G_{\gamma}(V)=\frac{1}{m}W$.
\end{center}

\vspace{0,3cm}

Thus,  $\max\{\rank((\gamma_1)_A),\rank((\gamma_1)_B)\}=$
\begin{center}
\hspace{6cm}$=\max\{\rank(V),\rank(W)\}=\max\{\rank(R),\rank(S)\}=m$.
\end{center}

Moreover, notice that \begin{center}
$(\gamma_2)_A=F_{\gamma_2}(Id)=\frac{k}{k-m}F_{\gamma}(W^{\perp})=\frac{1}{m}V^{\perp}$\ \ and\ \  $(\gamma_2)_B=G_{\gamma_2}(Id)=\frac{k}{k-m}G_{\gamma}(V^{\perp})=\frac{1}{k-m}W^{\perp}$.
\end{center}

\vspace{0,3cm}

Therefore
$\max\{\rank((\gamma_2)_A),\rank((\gamma_2)_B)\}=\max\{\rank(V^{\perp}),\rank(W^{\perp})\}=k-m$. 

\vspace{0,3cm}

By their definitions, $\gamma_1$ and $\gamma_2$ are PPT. So, by lemma \ref{lemmaineqPPT}, 
$\rank(\gamma_1)\geq m$ and $\rank(\gamma_2)\geq k-m$.

\vspace{0,3cm}

Recall that $k=\rank(\gamma)=\rank(\gamma_1)+\rank(\gamma_2)\geq m+(k-m)$. Thus $\rank(\gamma_1)= m$ and $\rank(\gamma_2)= k-m$.

\vspace{0,3cm}

Since $\gamma=\frac{m}{k}\gamma_1+\frac{k-m}{k}\gamma_2$, $\gamma$ has $k$ eigenvalues equal to $\frac{1}{k}$ and $\gamma_1\gamma_2=0$, 

\begin{itemize}
\item $\gamma_1$ has $m$ eigenvalues equal to $\frac{1}{m}$ and the others $0$,
\item $\gamma_2$ has $k-m$ eigenvalues equal to $\frac{1}{k-m}$ and the others $0$.

\end{itemize}

\vspace{0,3cm}

Hence, \begin{itemize}
\item  $\gamma_1$ has $m$ eigenvalues equal to $\frac{1}{m}$, $(\gamma_1)_A=\frac{1}{m}V$, $(\gamma_1)_B=\frac{1}{m}W$ and $\rank(V)=\rank(W)=m$.
\item $\gamma_2$ has $k-m$ eigenvalues equal to $\frac{1}{k-m}$, $(\gamma_2)_A=\frac{1}{k-m}V^{\perp}$, $(\gamma_1)_B=\frac{1}{k-m}W^{\perp}$ and $\rank(V^{\perp})=\rank(W^{\perp})=k-m$.
\end{itemize}

\vspace{0,3cm}
By induction hypothesis,  $\gamma_1$ and $\gamma_2$ are separable and so is $\gamma$.\end{proof}

\vspace{0,5cm}

Finally, we prove the PPT counterpart of our last result.

\vspace{0,5cm}

\begin{theorem}
If $\gamma\in \mathcal{M}_k\otimes \mathcal{M}_k$ is a  PPT state such that $\rank(\gamma_A)=\rank(\gamma_B)=k$ then $\rank(\gamma)\geq k$.  In addition, if the equality holds then $\gamma$ is separable. 
\end{theorem}
\begin{proof}
We know already that $\rank(\gamma)\geq k$ by lemma \ref{lemmaineqPPT}. Let us assume that $\rank(\gamma)= k$.\vspace{0,2cm}

Define $\gamma_1=(Id\otimes S)\gamma(Id\otimes S^*)$ such that $(\gamma_1)_B=\frac{1}{k}Id$, where $S$ is invertible.\vspace{0,2cm}

Since $\gamma_1$ is also PPT, $\|\gamma_1\|_{\infty}\leq \|(\gamma_1)_B\|_{\infty}=\frac{1}{k}$, by lemma \ref{specialspectralradius}.\vspace{0,2cm}

Hence $ 1=tr((\gamma_1)_B)=tr(\gamma_1)\leq \|\gamma_1\|_{\infty}\rank(\gamma_1)=\frac{1}{k}.k=1$. So $\gamma_1$ has $k$ eigenvalues equal to $\frac{1}{k}$  and the others $0$.\vspace{0,2cm}

Notice that  $(F\overline{\gamma_1}F)^{\Gamma}$ is positive semidefinite and so is $(\gamma_1*F\overline{\gamma_1}F)^{\Gamma}=\gamma_1*(F\overline{\gamma_1}F)^{\Gamma}$ as a $*-$product of two positive semidefinite Hermitian matrices by item $a)$ of  remark \ref{remarkproduct}. So the positive semidefinite Hermitian matrix $\gamma_1*F\overline{\gamma_1}F$ is PPT. \vspace{0,2cm}

Now, by items $(8)$ and $(9)$ of lemma \ref{propertiesofrealignment}, $\mathcal{R}(\gamma_1*F\overline{\gamma_1}F)=\mathcal{R}(\gamma_1)\mathcal{R}(\gamma_1)^*$, which is positive semidefinite. \vspace{0,2cm}

Next, on one hand  $tr(\mathcal{R}(\gamma_1*F\overline{\gamma_1}F))=tr(\mathcal{R}(\gamma_1)\mathcal{R}(\gamma_1)^*)=tr(\gamma_1\gamma_1^*)=\frac{1}{k}$, since $\mathcal{R}$ is an isometry.\vspace{0,2cm}

On the other hand $\|\mathcal{R}(\gamma_1*F\overline{\gamma_1}F)^{\Gamma}\|_1=\|\mathcal{R}(\gamma_1*(F\overline{\gamma_1}F)F)\|_1$ (by item (6) of lemma \ref{propertiesofrealignment})

\vspace{0,2cm}

\hspace{6.8cm} $=\|\mathcal{R}(\gamma_1*(F\overline{\gamma_1}F)F)F\|_1$ (since $F$ is an isometry) \vspace{0,2cm}

\hspace{6.8cm} $=\|(\gamma_1*(F\overline{\gamma_1}F)^{\Gamma}\|_1$ (by item (4) of lemma \ref{propertiesofrealignment}) \vspace{0,2cm}

\hspace{6.8cm} $=tr(\gamma_1*(F\overline{\gamma_1}F))$ (since $\gamma_1*(F\overline{\gamma_1}F)$ is PPT) \vspace{0,2cm}

\hspace{6.8cm} $=tr((\gamma_1)_B(\gamma_1)_B^*)=\frac{1}{k}$ (by item $c)$ of remark \ref{remarkproduct}). \vspace{0,8cm}

Therefore, $tr(\mathcal{R}(\gamma_1*F\overline{\gamma_1}F))= \|\mathcal{R}(\gamma_1*F\overline{\gamma_1}F)^{\Gamma}\|_1$. \vspace{0,2cm}

Since  $\mathcal{R}(\gamma_1*F\overline{\gamma_1}F)$ is positive semidefinite, this last equality  means that 
 $\mathcal{R}(\gamma_1*F\overline{\gamma_1}F)^{\Gamma}$ is positive semidefinite, i.e, $\mathcal{R}(\gamma_1*F\overline{\gamma_1}F)$ is PPT.\vspace{0,2cm}

We  have just discovered that $\mathcal{R}(\gamma_1*F\overline{\gamma_1}F)$ and $\gamma_1*F\overline{\gamma_1}F$ are  PPT, but in this situation lemma \ref{lemmarealigmentPPTareequal} says  that $\gamma_1*F\overline{\gamma_1}F=\mathcal{R}(\gamma_1*F\overline{\gamma_1}F).$\vspace{0,2cm}

Now, by corollary \ref{corollaryleftfilter}, there is an invertible matrix $R\in \mathcal{M}_k$ such that
 $$\gamma_2=(R^*\otimes Id)\gamma_1(R\otimes Id)=\sum_{i=1}^n\lambda_i A_i\otimes B_i,$$

where
\begin{itemize}
\item[$a)$]  $\lambda_1\geq \lambda_i>0$ for every $i$ and $A_1=\frac{Id}{\sqrt{k}}$, 
\item[$b)$]  $A_i=A_i^*$, $B_i=B_i^*$  for every $i$,
\item[$c)$] $tr(A_iA_j)=tr(B_iB_j)=0$ for every $i\neq j$ and $tr(A_i^2)=tr(B_i^2)=1$.\vspace{0,2cm}
\end{itemize}

In addition, we can normalize its trace, so assume that $tr(\gamma_2)=1$. \vspace{0,2cm}

Hence, $\gamma_2*(F\overline{\gamma_2} F)=\sum_{i=1}^n\lambda_i^2 A_i\otimes \overline{A_i}=(R^*\otimes R^t)(\gamma_1*(F\overline{\gamma_1} F))(R\otimes \overline{R}).$\vspace{0,2cm}

Like $\gamma_1*(F\overline{\gamma_1} F)$, the positive semidefinite Hermitian matrix $\gamma_2*(F\overline{\gamma_2} F)$ is also invariant under realignment because \vspace{0,2cm}

$\mathcal{R}(\gamma_2*(F\overline{\gamma_2} F))=\mathcal{R}((R^*\otimes R^t)(\gamma_1*(F\overline{\gamma_1} F))(R\otimes \overline{R}))$

$\hspace{2.95cm}=(R^*\otimes R^t)\mathcal{R}(\gamma_1*(F\overline{\gamma_1} F))(R\otimes \overline{R})$, by item $(3)$ of lemma \ref{propertiesofrealignment},\vspace{0,2cm}

$\hspace{2.95cm}=(R^*\otimes R^t)(\gamma_1*(F\overline{\gamma_1} F))(R\otimes \overline{R})$, since $\gamma_1*F\overline{\gamma_1}F=\mathcal{R}(\gamma_1*F\overline{\gamma_1}F).$\vspace{0,2cm}

$\hspace{2.95cm}=\gamma_2*(F\overline{\gamma_2} F)$.\vspace{0,5cm}

Now, notice that $k\lambda_1^2=tr\left(\gamma_2*(F\overline{\gamma_2} F)\ (A_1\otimes \overline{A_1})\right)k$

\hspace{3,8cm} $=tr\left(\gamma_2*(F\overline{\gamma_2} F)\ (\frac{Id}{\sqrt{k}}\otimes \frac{Id}{\sqrt{k}})\right)k$\vspace{0,2cm}

\hspace{3,8cm} $=tr(\gamma_2*(F\overline{\gamma_2} F))$\vspace{0,2cm}

\hspace{3,8cm} $=tr(\mathcal{R}(\gamma_2*(F\overline{\gamma_2} F)))$, since $\gamma_2*F\overline{\gamma_2}F=\mathcal{R}(\gamma_2*F\overline{\gamma_2}F)$\vspace{0,2cm}

\hspace{3,8cm} $=tr(\mathcal{R}(\gamma_2)\mathcal{R}(\gamma_2)^*)$, by items (8--9) of lemma \ref{propertiesofrealignment} \vspace{0,2cm}

\hspace{3,8cm} $=tr(\gamma_2\gamma_2^*)$, since $\mathcal{R}$ is an isometry. \vspace{0,5cm}

Next, $\|\gamma_2\|^2_{\infty}\leq \|\mathcal{R}(\gamma_2)\|_{\infty}^2$, since $\gamma_2$ is PPT and lemma \ref{specialspectralradius}.\vspace{0,2cm}

Notice that the largest singular value of $G_{\gamma_2}$ is $\lambda_1$ by item $a)$ above and the definition of $G_{\gamma_2}$. Hence, by lemma \ref{lemmaoperatornormrealignment}, $\|\mathcal{R}(\gamma_2)\|_{\infty}=\lambda_1$. Moreover, 
remind that $\rank(\gamma_2)=\rank(\gamma_1)=\rank(\gamma)=k$. Therefore, 
\begin{center}
$k\lambda_1^2=tr(\gamma_2^2)\leq \|\gamma_2\|^2_{\infty}\rank(\gamma_2)\leq\lambda_1^2.k.$
\end{center}

\vspace{0,2cm}

The inequalities above are, in fact, equalities, which only happens  when all the $k$ non-null eigenvalues of $\gamma_2$ are equal to $\lambda_1$. Therefore, $1=tr(\gamma_2)=k\lambda_1$. So $\lambda_1=\frac{1}{k}$. In addition, 
$$1=tr(\gamma_2)=\lambda_1tr(A_1)tr(B_1)=\frac{1}{k}\sqrt{k}\ tr(B_1).$$

So $tr(B_1)=\sqrt{k}$ and $tr(B_1^2)=1$. Recall that $G_{\gamma_2}(\frac{1}{\lambda_1}A_1)=B_1$ is a positive semidefinite Hermitian matrix, since $G_{\gamma_2}$ is a positive map and $\frac{1}{\lambda_1}A_1=\frac{Id}{k\sqrt{k}}$. Under these conditions the only possibility for $B_1$ is  $B_1=\frac{Id}{\sqrt{k}}$. 

\vspace{0,2cm}
Finally, $\gamma_2$ has $k$ eigenvalues equal to $\frac{1}{k}$ and the others $0$ and \begin{center}
$(\gamma_2)_B=G_{\gamma_2}(Id)=G_{\gamma_2}(\sqrt{k}A_1)=\frac{Id}{k}$, \ \ \ $(\gamma_2)_A=F_{\gamma_2}(Id)=F_{\gamma_2}(\sqrt{k}B_1)=\frac{Id}{k}$. 
\end{center}

\vspace{0,1cm}

By lemma \ref{lemmaseparabilityPPT}, $\gamma_2$ is separable and so are $\gamma_1$ and $\gamma$.
\end{proof}
\vspace{0,2cm}
\section{Summary and Conclusion}
\vspace{0,2cm}
In this article we proved new results for a triad of types of quantum states which includes the positive under partial transpose type. We obtained the same upper bound for the spectral radius of  these types of quantum states. Then we showed that two of these types can be put in the filter normal form retaining their shapes. Finally, we proved that there is a lower bound for their ranks and whenever this lower bound is attained these states are separable. This last result is another consequence of their complete reducibility property. This is plenty of evidence that these states are deeply connected. In addition, their complete reducibility property is a unifying force connecting and providing many results in entanglement theory. 
\vspace{0,2cm}

\section{Disclosure Statement}
\vspace{0,2cm}
No potential conflict of interest was reported by the author.
\vspace{0,2cm}

\end{document}